\documentclass[10pt, a4paper]{article}
\pdfoutput=1
\usepackage[footskip=35.0pt,textheight=625pt]{geometry}
\usepackage[utf8]{inputenc}
\usepackage[T1]{fontenc}
\usepackage[colorlinks, pagebackref=true,bookmarks=false]{hyperref}

\hypersetup{
  linkcolor=[rgb]{0.3,0.3,0.6},
  citecolor=[rgb]{0.2, 0.6, 0.2},
  urlcolor=[rgb]{0.6, 0.2, 0.2}
}

\usepackage{microtype}
\usepackage[english]{babel}
\usepackage{enumitem}
\usepackage{booktabs}
\usepackage{listings}
\usepackage{titlesec}
\titlelabel{\thetitle.\quad}
\titleformat*{\subsubsection}{\bfseries}
\titleformat*{\paragraph}{\bfseries}

\usepackage{float}

\usepackage{mathtools, amsthm, amsfonts, amssymb, commath}

\usepackage{tikz}
\usetikzlibrary{arrows.meta}
\usepackage{accents}

\DeclarePairedDelimiter\ceil{\lceil}{\rceil}

\newcommand{\cw}{\operatorname{cw}}
\newcommand{\CW}{\operatorname{CW}}

\newcommand{\eps}{\varepsilon}

\newcommand{\FF}{\mathbb{F}}
\newcommand{\RR}{\mathbb{R}}
\newcommand{\NN}{\mathbb{N}}
\newcommand{\prob}{\mathcal{P}}

\DeclareMathOperator{\supp}{supp}

\DeclareMathOperator{\rank}{R}

\DeclareMathAccent{\wtilde}{\mathord}{largesymbols}{"65}
\DeclareMathOperator{\asymprank}{\underaccent{\wtilde}{R}}
\DeclareMathOperator{\asympsubrank}{\underaccent{\wtilde}{Q}}

\DeclareMathOperator{\Oh}{\mathcal{O}}

\newcommand{\degenleq}{\unlhd}

\usepackage{thmtools, thm-restate} 
\declaretheorem[name=Theorem, parent=section]{theorem}
\declaretheorem[name=Corollary, sibling=theorem]{corollary}

\declaretheorem[name=Lemma, sibling=theorem]{lemma}

\theoremstyle{definition}
\declaretheorem[name=Definition, sibling=theorem]{definition}
\declaretheorem[name=Remark, sibling=theorem]{remark}

\theoremstyle{remark}

\title{Barriers for Rectangular Matrix Multiplication}
\author{Matthias Christandl, François Le Gall, Vladimir Lysikov and Jeroen Zuiddam}

\begin{document}

\maketitle

{\par\noindent
\textbf{Abstract.}
We study the algorithmic problem of multiplying large matrices that are rectangular.
We prove that the method that has been used to construct the fastest algorithms for rectangular matrix multiplication cannot give algorithms with complexity $n^{p + 1}$ for $n \times n$ by $n \times n^p$ matrix multiplication. In fact, we prove a precise numerical barrier for this method. Our barrier improves the previously known barriers, both in the numerical sense, as well as in its generality. 
In particular, we prove that any lower bound on the dual exponent of matrix multiplication $\alpha$ via the big Coppersmith--Winograd tensors cannot exceed 0.6218.
}

\tableofcontents

\section{Introduction}

Given two large matrices, how many scalar arithmetic operations (addition, subtraction and multiplication) are required to compute their matrix product?

The standard algorithm for multiplying two square matrices of shape $n\times n$ costs roughly~$2n^3$ arithmetic operations. On the other hand, we know that at least~$n^2$ operations are required. Denoting by $\omega$ the optimal exponent of $n$ in the number of operations required by any arithmetic algorithm, we thus have $2 \leq \omega \leq 3$. What is the value of $\omega$? Since Strassen published his matrix multiplication algorithm in 1969 we know that $\omega \leq 2.81$~\cite{strassen-gaussian-1969}. Over the years, more constructions of faster matrix multiplication algorithms, relying on insights involving direct sum algorithms, approximative algorithms and asymptotic induced matchings, led to the current upper bound $\omega \leq 2.371339$~\cite{DBLP:journals/jsc/CoppersmithW90, stothers2010complexity, MR2961552, le2014powers, DBLP:conf/soda/AlmanW21, DBLP:conf/soda/GallU18, DWZ23, DBLP:conf/soda/WilliamsXXZ24, alman2024asymmetryyieldsfastermatrix}.%

In applications, the matrices to be multiplied are often very rectangular instead of square; see the examples in~\cite{DBLP:conf/soda/GallU18} and below.
For any nonnegative real $p$, given an~$n \times \lceil n^p\rceil$ matrix and an $\lceil n^p \rceil \times n$ matrix, how many arithmetic operations are required to compute their product? Denoting, similarly as in the square case, by $\omega(p)$ the optimal exponent of $n$ in the number of operations required by any arithmetic algorithm\footnote{Formally speaking, $\omega(p)$ is the infimum over all real numbers $b$ so that the product of any~$n \times \lceil n^p\rceil$ matrix and any $\lceil n^p \rceil \times n$ matrix can be computed in $\Oh(n^{b})$ arithmetic operations.
$\omega$ is defined analogously with square matrix multiplication, so $\omega = \omega(1)$}, we a priori have the bounds $\max(2, 1+p) \leq \omega(p) \leq 2+p$. 
What is the value of~$\omega(p)$?
Parallel to the developments in upper bounding~$\omega$, the upper bound $2+p$ was improved drastically over the years for several regimes of~$p$ \cite{Huang+98,Ke+08,le2012faster,DBLP:conf/soda/GallU18,DBLP:conf/soda/WilliamsXXZ24, LG24,  alman2024asymmetryyieldsfastermatrix}. The best lower bound on~$\omega(p)$, however, has remained $\max(2, 1+p)$.

So the matrix multiplication exponent $\omega$ characterises the complexity of square matrix multiplication
and, for every nonnegative real $p$, the rectangular matrix multiplication exponent~$\omega(p)$ characterises the complexity of rectangular matrix multiplication.
Coppersmith~\cite{DBLP:journals/siamcomp/Coppersmith82} proved that there exists
a value $0 < p < 1$ such that~$\omega(p) = 2$.
The largest $p$ such that $\omega(p) = 2$ is denoted by $\alpha$. We will refer to $\alpha$ as the \emph{dual matrix multiplication exponent}.
The algorithms constructed in~\cite{DBLP:conf/soda/WilliamsXXZ24} give the currently best bound $\alpha > 0.321334$. If~$\alpha = 1$, then of course $\omega = 2$.
In fact, $\omega + \tfrac{\omega}{2}\alpha \leq 3$ (\autoref{alphaomega}).
Thus we study $\omega(p)$ not only to understand rectangular matrix multiplication, but also as a means to prove $\omega = 2$. The value of $\alpha$ appears explicitly in various applications, for example in the recent work on solving linear programs~\cite{10.1145/3313276.3316303, DBLP:conf/soda/Brand20} and empirical risk minimization~\cite{lee2019solving}.

The goal of this paper is to understand why current techniques have not closed the gap between the best lower and upper bound on $\omega(p)$, and to thus understand where to find faster rectangular matrix multiplication algorithms.
We prove a barrier for current techniques to give much better upper bounds than the current ones. %
Our work gives a very precise picture of the limitations of current techniques used to obtain the best upper bounds on $\omega(p)$ and the best lower bounds on $\alpha$.

Our ideas apply as well to $n\times \lceil n^p\rceil$ by $\lceil n^p \rceil \times \lceil n^q \rceil$ matrix multiplication for different $p$ and $q$. We focus on $p=q$ for simplicity.

\subsection{How are matrix multiplication algorithms constructed?}\label{how}
To understand what are the current techniques that we prove barriers for, we explain how the current fastest algorithms for matrix multiplication are constructed, on a high level.
An algorithm for matrix multiplication should be thought of as a reduction of the ``matrix multiplication problem'' to the natural ``unit problem'' that corresponds to multiplying numbers,
\[
\textnormal{matrix multiplication problem} \leq \textnormal{unit problem}.
\]
Mathematically, problems correspond to families of tensors.
Several different notions of reduction are used in this context. We will discuss tensors and reductions in more detail later.

Historically, the asymptotically fast matrix multiplication algorithms for square or rectangular matrices, are obtained by a reduction of the matrix multiplication problem to some intermediate problem and a reduction of the intermediate problem to the unit problem,
\[
\textnormal{matrix multiplication problem} \leq \textnormal{intermediate problem} \leq  \textnormal{unit problem}.
\]
The intermediate problems that have been used so far to obtain the best upper bounds on $\omega(p)$ correspond to the so-called small and big Coppersmith--Winograd tensors $\cw_q$ and $\CW_q$.

Depending on the intermediate problem and the notion of reduction, we prove a barrier on the best upper bound on $\omega(p)$ that can be obtained in the above way. Before we say something about our new barrier, we discuss the history of barriers for matrix multiplication.

\subsection{History of matrix multiplication barriers}

We call a lower bound for all upper bounds on $\omega$ or $\omega(p)$ that can be obtained by some method, a \emph{barrier} for that method. We give a high-level historical account of barriers for square and rectangular matrix multiplication.

Ambainis, Filmus and Le~Gall~\cite{MR3388238} were the first to prove a barrier in the context of matrix multiplication.
They proved that a variety of methods applied to the Coppersmith--Winograd intermediate tensors (which gave the current best upper bounds on $\omega$) cannot give $\omega = 2$ and in fact cannot give $\omega \leq 2.3$.

Alman and Vassilevska Williams~\cite{alman_et_al:LIPIcs:2018:8360, 8555139} proved barriers for a notion of reduction called monomial degeneration, extending the realm of barriers beyond the scope of the paper of Ambainis \emph{et al}. %
They prove that some collections of intermediate tensors, including the Coppersmith--Winograd
intermediate tensors, cannot be used to prove $\omega = 2$.
Their analysis is based on studying the so-called asymptotic independence number of the intermediate problem (also called monomial asymptotic subrank).
Their paper also for the first time studies barriers for \emph{rectangular} matrix multiplication, for $0\leq p\leq 1$ and monomial degeneration. For example, they prove that the intermediate tensor~$\CW_6$ can only give~$\alpha \leq 0.872$ \cite[Cor.~6.1 for $p=8$]{alman_et_al:LIPIcs:2018:8360}. 

Blasiak \emph{et al.}~\cite{MR3631613, blasiak2017groups} studied barriers for square matrix multiplication algorithms obtained with a subset of the group-theoretic method, which is a monomial degeneration applied to certain group algebra tensors.

Christandl, Vrana and Zuiddam~\cite{DBLP:conf/coco/ChristandlVZ19}
proved barriers that apply more generally than the previous one, namely for a type of reduction called degeneration.
Their barrier is given in terms of the irreversibility of the intermediate tensor.
Intuitively, irreversibility can be thought of as an asymptotic measure of the failure of Gaussian elimination to bring tensors into diagonal form.
To compute irreversibility, they used the asymptotic spectrum of tensors and in particular two families of real tensor parameters with special algebraic properties: the quantum functionals~\cite{DBLP:conf/stoc/ChristandlVZ18} and support functionals~\cite{strassen1991degeneration}, although one can equivalently use asymptotic slice rank to compute the barriers for the Coppersmith--Winograd intermediate tensors.
Alman~\cite{DBLP:conf/coco/Alman19} simultaneously and independently obtained the same barrier, relying on a study of asymptotic slice rank.

\subsection{New barriers for rectangular matrix multiplication}

We prove new barriers for rectangular matrix multiplication using a class of tensor parameters called adequate tensor parameters. These include the quantum functionals and support functionals.

We first set up a general barrier framework that encompasses all previously used notions of reductions and then numerically compute barriers for the degeneration notion of reduction and the Coppersmith--Winograd intermediate problems. We also discuss barriers for ``mixed'' intermediate problems, which covers a method used by, for example, Coppersmith~\cite{DBLP:journals/jc/Coppersmith97}.

We will explain our barrier in more detail in the language of tensors, but first we will give a numerical illustration of the barriers.

\subsubsection{Numerical illustration of the barriers}\label{subsub:fig}

For the popular intermediate tensor $\CW_6$ our barrier to get upper bounds on~$\omega(p)$ (for various $p$) via degeneration looks as follows. In \autoref{fig1}, the horizontal axis goes over all~$p \in [0,2]$. The blue line is the upper bound on $\omega(p)$ obtained via $\CW_6$ as in~\cite{le2012faster}.\footnote{Better upper bounds have been obtained in \cite{DBLP:conf/soda/WilliamsXXZ24}; for these high-level comparisons they do not change the general picture.} The yellow line is our barrier. The red line is the best lower bound $\max\{2, 1+p\}$ on~$\omega(p)$. (We note that, in~\cite{le2012faster}, the best upper bounds on $\omega(p)$ are obtained using $\CW_q$ with~$q=5$ for $p\leq0.81$, $q=6$ for $0.81<p\leq3.5$ and $q=7$ for $p>3.5$.)

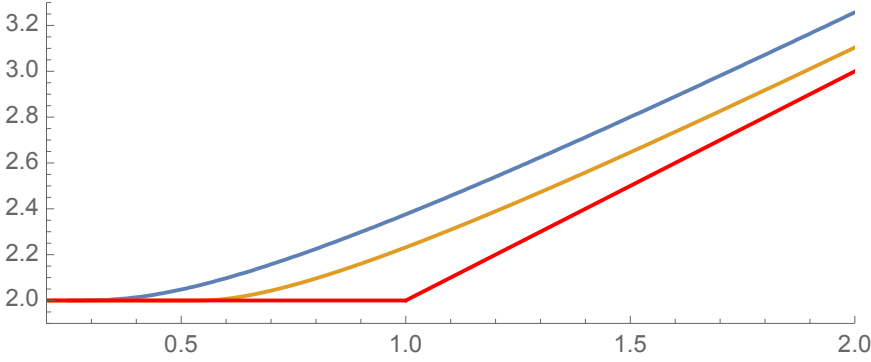
\begin{figure}[H]
\centering
\begin{tikzpicture}
    \node[anchor=south west,inner sep=0] at (0,0) {\includegraphics[scale=0.8]{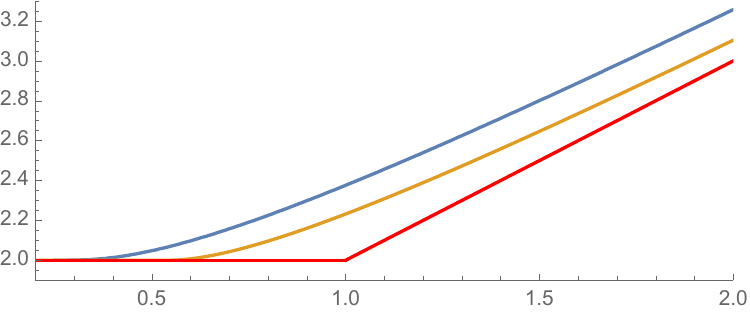}};
    \draw[color=gray,-latex] (0.48,0.5)--(10.5,0.5) node[below,color=black] {$p$};
    \draw[color=gray,-latex] (0.48,0.5)--(0.48,4.7) node[left,color=black] {$\omega$};
    \draw[fill,color={rgb,256:red,91;green,124;blue,175}] (10.6,4.4)--(10.8,4.4)--(10.8,4.6)--(10.6,4.6)--cycle;
    \node[anchor=west,scale=0.7] at (11,4.5) {upper bound~\cite{le2012faster}};
    \draw[fill,color={rgb,256:red,225;green,156;blue,36}] (10.6,3.9)--(10.8,3.9)--(10.8,4.1)--(10.6,4.1)--cycle;
    \node[anchor=west,scale=0.7] at (11,4.0) {our barrier};
    \draw[fill,color=red] (10.6,3.4)--(10.8,3.4)--(10.8,3.6)--(10.6,3.6)--cycle;
    \node[anchor=west,scale=0.7] at (11,3.5) {$\max \{2, 1 + p\}$};
\end{tikzpicture}
\caption{The blue line is the upper bound on $\omega(p)$ obtained via $\CW_6$ as in \cite{le2012faster} where $p \in [0,2]$ in on the horizontal axis. The yellow line is our barrier for upper bounds on~$\omega(p)$ via degeneration and the intermediate tensor $\CW_6$. The red line is the lower bound on~$\omega(p)$.}
\label{fig1}
\end{figure}

In~\autoref{fig2} we give the barrier values for $\CW_q$ for $q\in \{2, \ldots, 8\}$, in terms of the dual matrix multiplication exponent~$\alpha$. (We recall that $\alpha$ is the largest value of $p$ such that~$\omega(p) = 2$.) For $q=6$, this barrier value equals the smallest value of $p$ in \autoref{fig1} where the yellow line goes  above $2$.

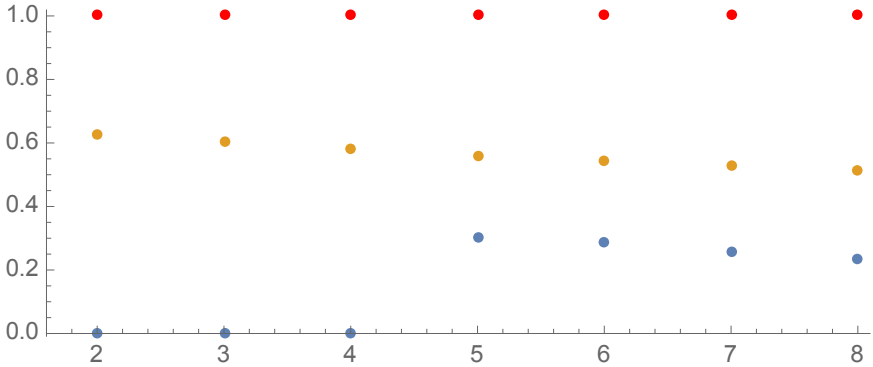
\begin{figure}[H]
\centering
\begin{tikzpicture}
    \node[anchor=south west,inner sep=0] at (0,0) {\includegraphics[scale=0.8]{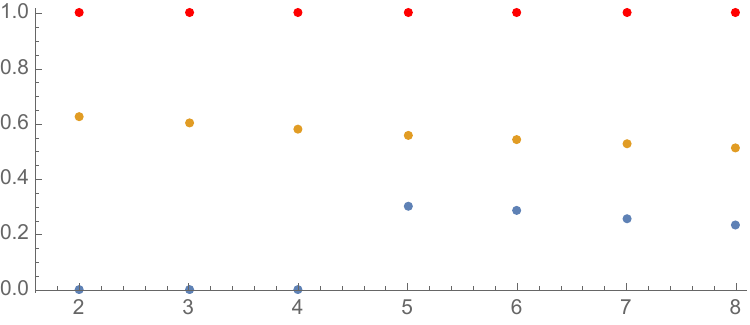}};
    \draw[color=gray,-latex] (0.48,0.48)--(10.5,0.48) node[below,color=black] {$q$};
    \draw[color=gray,-latex] (0.48,0.48)--(0.48,4.8) node[left,color=black] {$\alpha$};
    \draw[fill,color={rgb,256:red,91;green,124;blue,175}] (10.6,4.4)--(10.8,4.4)--(10.8,4.6)--(10.6,4.6)--cycle;
    \node[anchor=west,scale=0.7] at (11,4.5) {lower bound \cite{le2012faster}};
    \draw[fill,color={rgb,256:red,225;green,156;blue,36}] (10.6,3.9)--(10.8,3.9)--(10.8,4.1)--(10.6,4.1)--cycle;
    \node[anchor=west,scale=0.7] at (11,4.0) {our barrier};
    \draw[fill,color=red] (10.6,3.4)--(10.8,3.4)--(10.8,3.6)--(10.6,3.6)--cycle;
    \node[anchor=west,scale=0.7] at (11,3.5) {upper bound ($=\!1$)};
\end{tikzpicture}
\caption{The blue points are the lower bounds on $\alpha$ obtained via $\CW_q$ as in~\cite{le2012faster} for all $q \in \{2, \ldots, 8\}$. The yellow points are our barriers for the best lower bound on $\alpha$ obtainable via degeneration and the intermediate tensor $\CW_q$. The red points are the best upper bounds on $\alpha$, namely $1$. The lower bound $\alpha>0.3029$ in \cite{le2012faster} is attained using $q=5$. Any lower bound on $\alpha$ using degeneration and $\CW_q$ for any $q$, cannot exceed $0.6218$, the highest yellow point in the graph.}
\label{fig2}
\end{figure}

Our results give that the best lower bound on $\alpha$ obtainable with degenerations via $\CW_q$ for any $q$, cannot exceed $0.6218$. (This value corresponds to the highest yellow point in \autoref{fig2}. See also \autoref{subsec:values}.)  Recall that the currently best lower bound is~$\alpha > 0.321334$~\cite{DBLP:conf/soda/WilliamsXXZ24}.

Compared to \cite{alman_et_al:LIPIcs:2018:8360} our barriers are more general, numerically higher and apply not only for $0\leq p\leq 1$ but also for $p \geq 1$. For example, \cite{alman_et_al:LIPIcs:2018:8360} proves that monomial degeneration via $\CW_6$ can only give~$0.871 \leq \alpha$ whereas we get that the stronger degenerations via $\CW_6$ can only give~$0.543 \leq \alpha$.

\subsubsection{The barrier in tensor language}

Let us continue the discussion that we started in \autoref{how} of how algorithms are constructed, but now in the language of tensors. The goal is to explain our barrier in more detail.

As we mentioned, algorithms correspond to reductions from the matrix multiplication problem to some natural unit problem and the problems correspond to tensors.
Let $\FF$ be some fixed base field. (The value of $\omega(p)$ may in fact depend on the characteristic of the base field.)
A tensor is a trilinear map $\FF^{n_1} \times \FF^{n_2} \times \FF^{n_3} \to \FF$. The problem of multiplying an~$\ell \times m$ matrix and an $m \times n$ matrix corresponds to the matrix multiplication tensor
\[
\langle \ell, m, n\rangle = \sum_{i = 1}^\ell \sum_{j = 1}^m \sum_{k = 1}^n x_{ij} y_{jk} z_{ki}.
\]
The unit problem corresponds to the family of diagonal tensors
\[
\langle n\rangle = \sum_{i = 1}^n  x_i y_i z_i.
\]
There are several notions of reduction that one can consider, but the following is the most natural one. For two tensors $S$ and $T$ we say $S$ is a restriction of $T$ and write~$S \leq T$ if there are three linear maps $A,B,C$ of appropriate formats such that~$S$ is obtained from~$T$ by precomposing with $A$, $B$ and $C$, that is, $S = T \circ (A,B,C)$.

A very important observation (see, e.g., \cite{burgisser1997algebraic} or \cite{blaser2013fast}) is that any matrix multiplication algorithm corresponds to a tensor restriction
\[
\langle \ell, m, n\rangle \leq \langle r\rangle.
\]
Square matrix multiplication algorithms look like
\[
\langle n, n, n\rangle \leq \langle r\rangle
\]
and rectangular matrix multiplication, of the form that we study, look like
\[
\langle n, n, \lceil n^p\rceil\rangle \leq \langle r\rangle.
\]
In general, faster algorithms correspond to having smaller $r$ on the right-hand side. In fact, if
\[
\langle n, n, n\rangle \leq \langle n^{c + o(1)}\rangle
\]
then $\omega \leq c$, and similarly for any $p \geq 0$, if
\[
\langle n, n, \lceil n^p\rceil\rangle \leq \langle n^{c + o(1)}\rangle
\]
then $\omega(p) \leq c$. For example, if
\[
\langle n, n, n^3\rangle \leq \langle n^{c + o(1)}\rangle
\]
then $\omega(3) \leq c$.

Next we utilise a natural product structure on matrix multiplication tensors which is well known as the fact that block matrices can be multiplied block-wise. For tensors $S$ and $T$ one naturally defines a Kronecker product $S \otimes T$ generalizing the matrix Kronecker product. Then the matrix multiplication tensors multiply like $\langle n_1, n_2, n_3\rangle \otimes \langle m_1, m_2, m_3\rangle = \langle n_1m_1, n_2m_2, n_3m_3\rangle$ and the diagonal tensors multiply like~$\langle n\rangle \otimes \langle m\rangle = \langle nm\rangle$.

We can thus say: if
\[
\langle 2, 2, 2^3\rangle^{\otimes n} \leq \langle 2\rangle^{\otimes c n + o(n)}
\]
then $\omega(3) \leq c$. We now think of our problem as the problem of determining the optimal asymptotic rate of transformation from $\langle 2\rangle$ to $\langle 2,2,2^3\rangle$. Of course we can do similarly for values of $p$ other than $p=3$, if we deal carefully with $p$ that are non-integer. For clarity we will in this section stick to $p=3$.

In practice, as mentioned before, algorithms are obtained by reductions via intermediate problems. This works as follows. Let $T$ be any tensor, the intermediate tensor. Then clearly, if
\begin{equation}\label{eq:method}
\langle 2, 2, 2^3\rangle^{\otimes n} \leq T^{\otimes a n + o(n)} \leq  \langle 2\rangle^{\otimes ab n + o(n)},
\end{equation}
then $\omega(3) \leq a b$. The barrier we prove is a lower bound on $a b$ depending on $T$ and the notion of reduction used in the inequality $\langle 2, 2, 2^3\rangle^{\otimes n} \leq T^{\otimes a n + o(n)}$, which in this section we take to be restriction.

We obtain the barrier as follows. Suppose that $F$ is a map from the set of tensors to the nonnegative real numbers that is $\leq$-monotone, $\otimes$-multiplicative and $\langle n\rangle$-normalised, meaning that for any tensors $S$ and $T$ the following holds: if $S \leq T$ then~$F(S) \leq F(T)$; $F(S\otimes T) = F(S) F(T)$ and $F(\langle n\rangle) = n$.  (These conditions on $F$ can be slightly weakened, which we will do in a moment.)
We apply $F$ to both sides of the first inequality in \eqref{eq:method} to get
\[
F(\langle 2,2,2^3\rangle) \leq F(T)^a
\]
and so
\[
\frac{\log F(\langle 2,2,2^3\rangle)}{\log F(T)} \leq a
\]
Let $G$ be another map from tensors to reals that is $\leq$-monotone, $\otimes$-multiplicative and $\langle n\rangle$-normalised. We apply $G$ to both sides of the second inequality in \eqref{eq:method} to get
\[
G(T) \leq 2^b
\]
and so
\[
\log G(T) \leq b.
\]
We conclude that
\[
\frac{\log F(\langle 2,2,2^3\rangle)}{\log F(T)} \log G(T) \leq a b.
\]
Our barrier is thus
\[
\max_{F, G}\frac{\log F(\langle 2,2,2^3\rangle)}{\log F(T)} \log G(T) \leq a b.
\]
where the maximisation is over the $\leq$-monotone, $\otimes$-multiplicative and $\langle n\rangle$-normalised maps from tensors to reals. 

Let us now discuss suitable choices for the maps $F$ and $G$. Since $\max_G G(T)$ equals the asymptotic rank $\asymprank(T) = \lim_{n\to\infty} \rank(T^{\otimes n})^{1/n}$ by asymptotic spectrum duality \cite{Strassen:88-asymptotic}, we may write the barrier as
\[
\max_{F}\frac{\log F(\langle 2,2,2^3\rangle)}{\log F(T)} \log \asymprank(T) \leq a b.
\]
The asymptotic rank $\asymprank(T)$ we generally do not know how to compute. The best lower bounds we have are the flattening ranks, which are simply the matrix rank of the matrix obtained by grouping together two of the three tensor legs of $T$ (in one of three possible ways).

Regarding the choice of maps $F$, 
for tensors over the complex numbers, we know a family of $\leq$-monotone, $\otimes$-multiplicative and $\langle n\rangle$-normalised maps from tensors to reals, called the quantum functionals~\cite{DBLP:conf/stoc/ChristandlVZ18}.
To make our results more general, we will carry out the above reasoning to obtain the barrier using a larger class of maps that we call \emph{adequate maps} (which we will discuss later, \autoref{def:adequate}).
For tensors over the complex numbers, the quantum functionals are adequate.\footnote{Generally, all elements in the asymptotic spectrum of tensors \cite{Strassen:88-asymptotic} are adequate maps.}
For tensors over any field, a family of adequate maps is known, called the support functionals~\cite{strassen1991degeneration}. Our main barrier result then reads as follows:

\begin{theorem}\label{th:intro:main}
Upper bounds on $\omega(p)$ obtained via the intermediate tensor $T$ are at least
\[
\max_{F}\frac{\log(F(\langle 2,1,1\rangle)F(\langle 1,2,1\rangle)F(\langle 1,1,2\rangle)^p)}{\log F(T)} \log \asymprank(T),
\]
where the maximisation is over all adequate maps. %
\end{theorem}

See \autoref{thm:approx} for the precise statement of the result and \autoref{subsub:fig} for illustrations.

For the dual exponent $\alpha$ we prove the following barrier (precise statement in \autoref{th:barrier-alpha}).

\begin{theorem}\label{th:barrier-alpha-intro}
    For any $0 < p < 1$, is $T$ is used as an intermediate tensor to prove that $p$ is a lower bound on $\alpha$, then
    \[
	p \leq \min_F \frac{2 \log F(T)}{\log \asymprank(T) \log F(\left<1,1,2\right>)} - \frac{\log F(\left<2,2,1\right>)}{\log F(\left<1,1,2\right>)},
    \]
    where the minimization is over all adequate maps $F$ such that $\log F(\left<1,1,2\right>) \neq 0$.
\end{theorem}

In \autoref{sec:numerical} we will use the support functionals to obtain concrete numerical barriers for specific $T$ using \autoref{th:intro:main} and \autoref{th:barrier-alpha-intro}.

\begin{remark}
In \cite{DBLP:conf/coco/ChristandlVZ19} it was shown that any upper bound on the square matrix multiplication exponent $\omega = \omega(1)$ obtained via the intermediate tensor $T$ is at least
\[
2 \frac{\log \asymprank(T)}{\log \asympsubrank(T)}.
\]
This barrier can easily be recovered from \autoref{th:intro:main}. Indeed, when we set $p = 1$, the barrier in \autoref{th:intro:main} simplifies as follows. For any adequate $F$ (\autoref{def:adequate}) one can show that $F(\langle 2,1,1\rangle)F(\langle 1,2,1\rangle)F(\langle 1,1,2\rangle) = F(\langle 2,2,2\rangle)$ and $F(\langle 2,2,2\rangle) \geq \asympsubrank(\langle 2,2,2\rangle) = 4$, where~$\asympsubrank$ denotes the asymptotic subrank. Moreover, $\min_F F(T) = \asympsubrank(T)$ by asymptotic spectrum duality~\cite{Strassen:88-asymptotic}. Thus
\[
\max_{F}\frac{\log(F(\langle 2,1,1\rangle)F(\langle 1,2,1\rangle)F(\langle 1,1,2\rangle))}{\log F(T)} \log \asymprank(T) 
\geq \max_{F}\frac{2}{\log F(T)} \log \asymprank(T) = 2\frac{\log \asymprank(T)}{\log \asympsubrank(T)},
\]
which gives the claim.
There are several elements that make proving  \autoref{th:intro:main} more involved than the simpler barrier for the square matrix multiplication exponent of \cite{DBLP:conf/coco/ChristandlVZ19,DBLP:conf/coco/Alman19}. First of all, the barrier in \autoref{th:intro:main} makes more subtle use of the (adequate) maps $F$, namely their asymmetric nature and the fact that inside the maximization they appear in a numerator and denominator, leading to a much more interesting optimization problem. Moreover, contrary to our earlier running example in which we looked at the matrix multiplication tensor $\langle 2,2,2^3\rangle$, proving  \autoref{th:intro:main} involves considering the ``tensor'' $\langle 2,2,2^p\rangle$ for any real $p\geq0$. For this to make sense, we introduce (using adequate maps and limits) a new notion of a ``virtual matrix multiplication tensor'' (\autoref{subsec:virtual}), which will play a crucial role in our proofs. %
\end{remark}

\subsubsection{Catalyticity in matrix multiplication algorithms} We discussed that, in practice, the best upper bound on, say, $\omega(3)$ is obtained by a chain of inequalities of the form
\begin{equation}\label{eq:2}
\langle 2, 2, 2^3\rangle^{\otimes n} \leq T^{\otimes a n + o(n)} \leq  \langle 2\rangle^{\otimes ab n + o(n)}.
\end{equation}
We utilised this structure to obtain the barrier.
A closer look reveals that the methods used in practice have even more structure. Namely, they give an inequality that also has diagonal tensors on the left-hand side:
\begin{equation}\label{eq:3}
\langle2\rangle^{\otimes c n} \otimes \langle 2, 2, 2^3\rangle^{\otimes n} \leq T^{\otimes a n + o(n)} \leq  \langle 2\rangle^{\otimes ab n + o(n)}.
\end{equation}
The reason we say \eqref{eq:3} has more structure than \eqref{eq:2} is that \eqref{eq:3} implies a restriction of the form \eqref{eq:2} via recursive application.

Part of the tensor $\langle 2\rangle^{\otimes ab n + o(n)}$ on the far right-hand side acts as a catalyst since~$\langle2\rangle^{\otimes c n}$ is returned on the far left-hand side. We obtain better barriers when we have a handle on the amount of catalyticity $c$ that is used in the method~(see the schematic \autoref{fig3}), again by applying maps $F$ and $G$ to both sides of the two inequalities and deducing a lower bound on $ab$. The precise statement appears in \autoref{thm:approx}.

\begin{figure}[H]
\centering
\begin{tikzpicture}
    \node[anchor=south west,inner sep=0] at (0,0) {\includegraphics[scale=0.8]{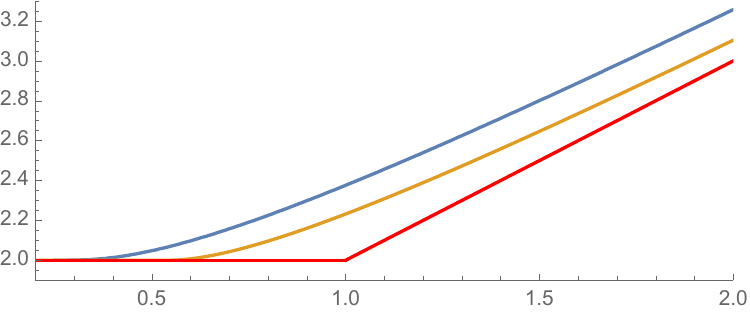}};
    \draw[color=gray,-latex] (0.48,0.5)--(10.5,0.5) node[below,color=black] {$p$};
    \draw[color=gray,-latex] (0.48,0.5)--(0.48,4.7) node[left,color=black] {$\omega$};
    \draw[fill,color={rgb,256:red,91;green,124;blue,175}] (10.6,4.4)--(10.8,4.4)--(10.8,4.6)--(10.6,4.6)--cycle;
    \node[anchor=west,scale=0.7] at (11,4.5) {upper bound~\cite{le2012faster}};
    \draw[fill,color={rgb,256:red,225;green,156;blue,36}] (10.6,3.9)--(10.8,3.9)--(10.8,4.1)--(10.6,4.1)--cycle;
    \node[anchor=west,scale=0.7] at (11,4.0) {our barrier};
    \draw[fill,color=red] (10.6,3.4)--(10.8,3.4)--(10.8,3.6)--(10.6,3.6)--cycle;
    \node[anchor=west,scale=0.7] at (11,3.5) {$\max \{2, 1 + p\}$};
    \draw[black,thick,-latex] (5.8,1.85) -> (5.8,2.15);
    \draw[black,thick,-latex] (6.8,2.28) -> (6.8,2.60);
    \draw[black,thick,-latex] (7.8,2.73) -> (7.8,3.03);
\end{tikzpicture}
\caption{This is the graph from \autoref{fig1} with arrows that indicate the influence of catalyticity. Roughly speaking, the barrier for $\CW_6$ (the yellow line) moves upwards when more catalyticity is used.}
\label{fig3}
\end{figure}

\subsection{Overview of the next sections}

In \autoref{sec:algs} we discuss in more detail the methods that are used to construct rectangular matrix multiplication algorithms and the different notions of reduction.

In \autoref{sec:barriers} we introduce and prove our barriers in the form of a general framework, dealing formally with non-integer $p$. We also discuss how to analyse ``mixed'' intermediate tensors.

In \autoref{sec:numerical} we discuss how to compute the barriers explicitly using the support functionals and we compute them for the Coppersmith--Winograd tensors $\CW_q$.

\section{Rectangular matrix multiplication algorithms}\label{sec:algs}

At the core of the methods that give the best upper bounds on the rectangular matrix multiplication exponent $\omega(p)$ lies the following theorem,
which can be proven using the asymptotic sum inequality for
rectangular matrix multiplication~\cite{DBLP:journals/tcs/LottiR83}
and the monotonicity of $\omega(p)$.

Denote by $\oplus$ the naturally defined direct sum for tensors. The rank $\rank(T)$ of a tensor $T$ is the smallest number $n$ such that $T \leq \langle n\rangle$, or equivalently, the smallest number $n$ such that $T(x,y,z) = \sum_{i=1}^n u_i(x)v_i(y)w_i(z)$ where $u_i,v_i,w_i$ are linear. The asymptotic rank $\asymprank(T)$ is defined as the limit $\lim_{n\to\infty} \rank(T^{\otimes n})^{1/n}$, which equals the infimum $\inf_n \rank(T^{\otimes n})^{1/n}$ since tensor rank is submultiplicative under $\otimes$ and bounded (using Fekete's lemma).

\begin{theorem}\label{thm:tau}
Let $m \geq n^p$.
If $\asymprank(\left< n, n, m \right>^{\oplus s}) \leq r$, then $s\, n^{\omega(p)} \leq r$.
\end{theorem}

Equivalently, phrased in the language of the introduction, \autoref{thm:tau} says that for any $m,n \in \NN$ such that $m \geq n^p$, if $\langle s\rangle^{\otimes k} \otimes \langle n,n,m \rangle^{\otimes k} \leq \langle r\rangle^{\otimes  k + o(k)}$
then $s n^{\omega(p)} < r$.  In practice, the upper bound $\asymprank(\left< n, n, m \right>^{\oplus s}) \leq r$ is obtained from a restriction $\langle s\rangle^{\otimes k} \otimes \left< n, n, m \right>^{\otimes k} \leq T^{\otimes a k + o(k)}$ for some intermediate tensor $T$  and an upper bound on the asymptotic rank~$\asymprank(T)$. In \autoref{sec:barriers} we will prove barriers for algorithms obtained in this way.

\subsubsection*{Reductions}
The restriction in the aforementioned inequality $\langle s\rangle^{\otimes k} \otimes \left< n, n, m \right>^{\otimes k} \leq T^{\otimes a k + o(k)}$ may be (and has been, in the literature) replaced by other types of reductions, which we will now discuss. (See also \cite{Strassen:87-relative, burgisser1997algebraic,blaser2013fast}.)

\emph{Degeneration} is a very general reduction that extends restriction.
Let $T : \FF^{n_1} \times \FF^{n_2} \times \FF^{n_3} \to \FF$ and $S : \FF^{m_1} \times \FF^{m_2} \times \FF^{m_3} \to \FF$ be trilinear maps.
We say $S$ is a degeneration of $T$ and write $S \degenleq T$ if $S = \lim_{\eps\to 0} T(A(\eps)x, B(\eps)y, C(\eps)z)$
for some matrices $A(\eps), B(\eps), C(\eps)$ with coefficients that are Laurent polynomials in $\eps$. %
Restriction $\leq$ defined above is the special case of degeneration where the matrices $A, B, C$ do not depend on $\eps$. %

There are also more restrictive notions of reductions which are easy to analyze combinatorially.
We say that $S$ is a \emph{monomial restriction} of $T$ and write $S \leq_M T$ if~$S = T(Ax, By, Cz)$ where the matrices of $A$, $B$ and $C$ have at most one nonzero entry in each row and column.
Essentially, $S$ is obtained from $T$ by rescaling some variables of the trilinear form and setting some of the variables to zero.

Similarly, we say that $S$ is a \emph{monomial degeneration} of $T$ and write $S \degenleq_M T$ if~$S = \lim_{\eps\to 0} T(A(\eps)x, B(\eps)y, C(\eps)z)$ where matrices $A(\eps), B(\eps), C(\eps)$ contain in each row and each column only one nonzero element.
Without loss of generality, the nonzero elements can be assumed to be monomials in~$\eps$.
Strassen's application of the laser method uses monomial degenerations. Coppersmith and Winograd~\cite{DBLP:journals/jsc/CoppersmithW90} uses monomial restrictions
where the variables zeroed out are chosen using a certain combinatorial gadget
(a Salem--Spencer set).
Later work building on the Coppersmith--Winograd construction retain this structure.

\subsubsection*{Coppersmith--Winograd intermediate tensors}
All improvements on the rectangular matrix multiplication exponent $\omega(p)$ since Coppersmith and Winograd \cite{DBLP:journals/jsc/CoppersmithW90} use the so-called Coppersmith--Winograd tensors as intermediate tensors, which are defined as
\[
\CW_q(x, y, z) = x_0 y_0 z_{q + 1} + x_0 y_{q + 1} z_0 + x_{q + 1} y_0 z_0 + \sum_{i = 1}^q (x_0 y_i z_i + x_i y_0 z_i + x_0 y_i z_i)
\]
It is known, because of a border rank decomposition, that
$\asymprank(\CW_q) = q + 2$. The barriers that we discuss in \autoref{sec:barriers} we will numerically evaluate for the $\CW_q$ tensors in \autoref{sec:numerical}. We will there, as an extra example, also evaluate barriers for the so-called ``little'' Coppersmith--Winograd tensors, see \autoref{rem:small-cw}.

\subsubsection*{Mixed Coppersmith--Winograd tensors}

Coppersmith~\cite{DBLP:journals/jc/Coppersmith97} used a mixture of
$\CW_q$ tensors with different $q$'s to upper bound~$\omega(p)$. We will analyze this class of methods in~\autoref{sec:mixed}. The best upper bounds in~\cite{le2012faster,DBLP:conf/soda/GallU18} do not use such a mixture of different $q$'s.

\section{Barriers for rectangular matrix multiplication}\label{sec:barriers}

In this section we prove barriers for certain methods to prove upper bounds on the rectangular matrix multiplication exponent. We begin with two preliminary subsections in which we introduce the notion of an ``adequate'' tensor parameter and the notion of a virtual matrix multiplication tensor. These notions will play a crucial role in stating and proving the barriers.

\subsection{Adequate tensor parameters}
We introduce a general class of tensor parameters, called adequate tensor parameters, in terms of which our barriers will be given later. %

Recall that $\leq$ denotes restriction on tensors as defined in the introduction.\footnote{We remark that everything we discuss in this section also holds if restriction is replaced with degeneration,  monomial degeneration or monomial restriction.}

\begin{definition}\label{def:adequate}
Let $F : \{\textnormal{tensors}\} \to \RR_{\geq 0}$ be any map. We call $F$ \emph{adequate} if it satisfies the following properties:
\begin{enumerate}[label=\upshape(\roman*)]
\item $\leq$-monotone: $F(S) \leq F(T)$ whenever $S \leq T$;
\item $\otimes$-submultiplicative: $F(S \otimes T) \leq F(S) \cdot F(T)$;
\item MaMu-$\otimes$-multiplicative:\[ F(\left<\ell_1 \ell_2, m_1 m_2, n_1 n_2 \right>) = F(\left<\ell_1, m_1, n_1\right>)\cdot F(\left<\ell_2, m_2, n_2\right>);\]
\item self-$\oplus$-additive: $F(T^{\oplus s}) = s\cdot F(T)$;
\item bounded by the asymptotic rank $\asymprank$: $F(T) \leq \asymprank(T)$.
\end{enumerate}
\end{definition}

Two known families of adequate tensor parameters %
are the ``upper support functionals'' of Strassen~\cite{strassen1991degeneration} and the ``quantum functionals'' of Christandl, Vrana and Zuiddam~\cite{DBLP:conf/stoc/ChristandlVZ18} (which we will not explicitly use in this paper). 
In \autoref{sec:numerical} we will discuss and use the upper support functionals.
In the rest of this section we will work with the abstract notion of adequate tensor parameters (\autoref{def:adequate}).

\subsection{Virtual matrix multiplication tensors}\label{subsec:virtual}
For any $p \in \NN$ and adequate function $F$, we have a value $F(\langle2,2,2^p\rangle)$.
In this section we will extend $F(\langle2,2,2^p\rangle)$ to a continuous function in $p \in \RR_{\geq0}$. %
We first observe the following.

\begin{lemma}\label{lem:pow}
Suppose that $a$ and $b$ are positive integers and $p = \log_a b$.
If $m \geq n^p$, then for every adequate $F$ we have
\[
F(\left<n, n, m\right>) \geq F(\left<a, a, b\right>)^{\log_a n}.
\]
\end{lemma}
\begin{proof}
  For every rational number $\frac{s}{t} < \log_a n$ we have
  \begin{multline*}
    F(\left<n, n, m\right>) = F(\left<n, n, m\right>^{\otimes t})^{\frac{1}{t}} = F(\left<n^t, n^t, m^t\right>)^{\frac{1}{t}} \geq %
    F(\left<a^s, a^s, b^{s}\right>)^{\frac{1}{t}} =  F(\left<a,a,b\right>)^{\frac{s}{t}}. 
  \end{multline*}
  This proves the claim.
\end{proof}

From \autoref{lem:pow} it follows that $\log_a F(\left<a, a, a^p\right>)$ is
the same for any~$a$ with integer power $a^p$.
We introduce a notation for dealing with this
value without referring to the set of possible values of $a$

\begin{definition}
We introduce a formal symbol $\left<2, 2, 2^p\right>$
for each real $p \geq 0$, which we call a \emph{virtual matrix multiplication tensor}.
We extend adequate maps $F$ to virtual matrix multiplication tensors as follows.
If~$p = \log_a b$ for some positive integers $a$ and $b$, then we define
\[
F(\left<2, 2, 2^p\right>) = 2^{\log_a F(\left<a, a, b\right>)}.
\]
Otherwise, we define
\[
F(\left<2, 2, 2^p\right>) = \inf \{F(\langle2, 2, 2^P\rangle) \mid P \geq p,\, \exists a, b \in \mathbb{Z}_{\geq 0}\colon P = \log_a b\}.
\]
\end{definition}

If $p$ is integer, then the value of $F$ on $\langle 2,2,2^p\rangle$ as a tensor and as a virtual tensor coincide. %
Thus we identify the virtual matrix multiplication tensor
$\left<2, 2, 2^p\right>$ with the matrix multiplication tensor $\left<2, 2, 2^p\right>$ when the latter exists.

Using this notation, \autoref{lem:pow} can be rephrased as follows.
\begin{lemma}
  \label{lem:powers}
  If $m \geq n^p$, then $F(\left<n, n, m\right>) \geq F(\left<2, 2, 2^p\right>)^{\log n}$ for every adequate~$F$.
\end{lemma}

\begin{corollary}\label{cor:mon}
    For every adequate $F$ the function $p \mapsto F(\left<2, 2, 2^p\right>)$ is monotone.
\end{corollary}
\begin{proof}
    Let $0 \leq p < q$. If $q = \log_a b$ for some positive integers $a, b$, then 
    \[F(\left<2, 2, 2^p\right>) \leq F(\left<a, a, b\right>)^{\frac{1}{\log a}} = 2^{\log_a F(\left<a,a,b\right>)} = F(\left<2,2,2^q\right>).\]
    If $q$ is not an exact logarithm, then for every $Q \geq q$ such that $Q = \log_a b$ with integer $a, b$ we have $F(\left<2,2,2^p\right>) \leq F(\left<2,2,2^Q\right>)$ and therefore $F(\left<2,2,2^p\right>) \leq F(\left<2,2,2^q\right>)$ from the definition of $F(\left<2,2,2^q\right>)$ as an infinum.
\end{proof}

\begin{lemma}
    For every adequate $F$ the function $p \mapsto F(\left<2, 2, 2^p\right>)$ is continuous.
\end{lemma}
\begin{proof}
    Let $0 \leq p < q$. Choose integers $a, b, c$ such that $0 \leq \frac{a}{c} \leq p < q \leq \frac{b}{c}$ and $|\frac{b}{c} - \frac{a}{c}| \leq 2|q - p|$.
    By MaMu-multiplicativity of $F$ we have 
    \[F(\langle 2^c, 2^c, 2^b\rangle) = F(\langle2^c, 2^c, 2^a\rangle) F\left(\left<1,1,2\right>\right)^{b - a}.\]
    Note that $\frac{a}{c} = \log_{2^c}{2^a}$ and $\frac{b}{c} = \log_{2^c}{2^b}$. Therefore
    \begin{multline*}\log F(\langle2,2,2^{\frac{b}{c}}\rangle) - \log F(\langle2,2,2^{\frac{a}{c}}\rangle)\\ = \log_{2^c} F(\langle2^c, 2^c, 2^b\rangle) - \log_{2^c} F(\left<2^c, 2^c, 2^a\right>) = \frac{b - a}{c} \log_{2} F(\left<1, 1, 2\right>).\end{multline*}
    From the monotonicity of $F$ (\autoref{cor:mon}) it follows that
    \begin{multline*}
    \log F(\left<2,2,2^q\right>) - \log F(\left<2,2,2^p\right>) \leq \log F(\langle 2,2,2^{\frac{b}{c}}\rangle) - \log F(\langle2,2,2^{\frac{a}{c}}\rangle)\\ =
    \log_{2} F(\left<1, 1, 2\right>) \Bigl(\frac{b}{c} - \frac{a}{c}\Bigr) \leq 2\log_2 F(\left<1,1,2\right>)(q - p)
    \end{multline*}
    and thus the monotone function $p \mapsto \log F(\left<2,2,2^p\right>)$ is continuous (and in fact Lipschitz continuous).
\end{proof}

\begin{lemma}\label{quasi}
For any real $p\geq0$ and adequate $F$, 
\[
F(\left<2, 2, 2^p\right>) = F(\left<2, 1, 1\right>) F(\left<1, 2, 1\right>) F(\left<1, 1, 2\right>)^p.
\]
\end{lemma}
\begin{proof}
We have
  $F(\left<a, 1, 1\right>) = F(\left<2, 1, 1\right>)^{\log a}$
  because if $\log a \leq \frac{b}{c}$, then $a^c \leq 2^b$ and $F(\left<a, 1, 1\right>)^c \leq F(\left<2, 1, 1\right>)^b$,
  and if $\log a \geq \frac{b}{c}$, then $F(\left<a, 1, 1\right>)^c \geq F(\left<2, 1, 1\right>)^b$.
  Analogous results hold for $\left<1, a, 1\right>$ and $\left< 1, 1, a \right>$.

Suppose $p = \log_a b$. Then
  \begin{multline*}
    \log F(\left<2, 2, 2^p\right>) = \log_a F(\left<a, a, b\right>) = \log_a \left[ F(\left<a, 1, 1\right>) F(\left<1, a, 1\right>) F(\left<1, 1, b\right>) \right] \\= \log F(\left<2, 1, 1\right>) + \log F(\left<1, 2, 1\right>) + p \log F(\left<1, 1, 2\right>).
  \end{multline*}
  For arbitrary $p$ the result follows by a continuity argument.
\end{proof}

\begin{lemma}\label{lem:adeq-cont}
  If $m = n^{p + o(1)}$, then $\log_n F(\left<n, n, m\right>) = \log F(\left<2, 2, 2^p\right>) + o(1)$ for every adequate~$F$.
\end{lemma}
\begin{proof} We have
$F(\langle n,n,m\rangle) = F(\langle n,1,1\rangle) F(\langle1,n,1\rangle) F(\langle 1,1,m\rangle)$
and so
\begin{align*}
\log_n F(\langle n,n,m\rangle) &= \log F(\langle 2,1,1\rangle) + \log F(\langle 1,2,1\rangle) + \log_n(m) \log F(\langle 1,1,2\rangle)\\
&=\log F(\langle 2,2,2^p\rangle) + o(1) F(\langle 1,1,2\rangle),
\end{align*}
which proves the claim.
\end{proof}

\subsection{Barriers for $T$-methods}

For any tensor $T$ we define the notion of a $T$-method for upper bounds on~$\omega(p)$ as follows.

\begin{definition}[$T$-method] Suppose $\asymprank(T) \leq r$. Suppose we are given a collection of inequalities $\left<n, n, m\right>^{\oplus s} \leq T^{\otimes k}$ with $n^p \leq m$. Then \autoref{thm:tau} gives the upper bound~$\omega(p) \leq \hat{\omega}(p)$ where $\hat{\omega}(p) = \inf \{ k \log_n r - \log_n s \}$ where the infimum is taken over all $k, n, s$ appearing in the collection of inequalities. We then say $\hat{\omega}(p)$ is obtained by a \emph{$T$-method}.

We say that the $T$-method is \emph{$\kappa$-catalytic} if the set of values of $n$ is unbounded, the bound $\hat{\omega}(p)$ is not attained on any one reduction of the method (so $\hat{\omega}(p) = \lim \inf \{ k \log_n r - \log_n s \}$), and in any reduction we have $s \geq C n^\kappa$ for some constant~$C$.
\end{definition}

Note that for while for general $T$-methods we allow degenerate cases when the upper bound~$\hat{\omega}(p)$ is given already by one of the reductions of the method, we are mostly interested in methods where the upper bound appears as a limit for some sequence of reductions, so ``$\inf$'' in the definition of $\hat{\omega}(p)$ can be replaced by ``$\lim \inf$''. In particular, we require this behaviour for catalytic methods as it is used to obtain better barrier results in this case.

\begin{theorem}
  \label{thm:general}
  Any upper bound $\hat{\omega}(p)$ on $\omega(p)$ obtained by
  a $T$-method satisfies
  \begin{equation*}
    \hat{\omega}(p) \geq \frac{\log F(\left<2, 2, 2^p\right>) \log \asymprank(T)}{\log F(T)}
  \end{equation*}
  for every adequate $F$.
  
  Moreover, if the method is $\kappa$-catalytic, then %
  \begin{equation*}
    \hat{\omega}(p) \geq \frac{\log F(\left<2, 2, 2^p\right>) \log \asymprank(T)}{\log F(T)} + \kappa \left(\frac{\log \asymprank(T)}{\log F(T)} - 1\right).
  \end{equation*}
\end{theorem}
\begin{proof}
It is enough to prove the inequality for one reduction $T^{\otimes k} \geq \left<n, n, m\right>^{\oplus s}$ with $m \geq n^p$, which gives an upper bound $\hat{\omega}(p) = k \log_n \asymprank(T) - \log_n s$. %

  Using \autoref{lem:powers} and superadditivity of $F$, we have
  \[
    F(\left<n, n, m\right>^{\oplus s}) \geq s F(\left<n, n, m\right>) \geq s F(\left<2, 2, 2^p\right>)^{\log n}.
  \]
  Therefore $k \log_n F(T) \geq \log_n F(T^{\otimes k}) \geq \log F(\left<2, 2, 2^p\right>) + \log_n s$.
  For $\hat{\omega}(p)$ we get
  \begin{equation*}
    \frac{\hat{\omega}(p) + \log_n s}{\log F(\left<2, 2, 2^p\right>) + \log_n s} \geq \frac{k \log_n \asymprank(T)}{k \log_n F(T)} = \frac{\log \asymprank(T)}{\log F(T)}.
  \end{equation*}
Since $F(T) \leq \asymprank(T)$, we have $\hat{\omega}(p) + \log_n s \geq \log F(\left<2, 2, 2^p\right>) + \log_n s$ and therefore
  \begin{equation*}
    \frac{\hat{\omega}(p)}{\log F(\left<2, 2, 2^p\right>)} \geq \frac{\hat{\omega}(p) + \log_n s}{\log F(\left<2, 2, 2^p\right>) + \log_n s}.
  \end{equation*}
  If the method is $\kappa$-catalytic, then $\log_n s \geq \kappa + O(\frac{1}{\log n})$, and as $n \to \infty$ we have
  \begin{equation*}
    \frac{\hat{\omega}(p) + \kappa}{\log F(\left<2, 2, 2^p\right>) + \kappa} \geq \frac{\log \asymprank(T)}{\log F(T)}.
  \end{equation*}
This concludes the proof.
\end{proof}

\begin{remark}
    Note that in the definition of $\kappa$-catalytic we require that the infimum in $\hat\omega(p)$ is not a minimum. This is indeed what happens in the modern constructions of matrix multiplication algorithms. This requirement allows us to let $n$ go to infinity in the proof of \autoref{thm:general} to get rid of the $O(1/\log n)$ term.
\end{remark}

\subsection{Barriers for asymptotic $T$-methods}
To cover the methods that are used in practice we need the following notion.
\begin{definition}[Asymptotic $T$-method.] Let $T$ be a tensor. Suppose $\asymprank(T) \leq r$. Suppose we are given a collection of inequalities $\left<n, n, m\right>^{\oplus s} \leq T^{\otimes k}$ where the values of $n$ are unbounded and $m \geq f(n)$ for some function $f(n) = n^{p + o(1)}$. Then $\omega(p)$ is at most $\hat{\omega}(p)$ where $\hat{\omega}(p) = \lim \inf \{ k \log_n r - \log_n s \}$ where the limit is taken over all $k, n, s$ appearing in the collection of inequalities as $n \to \infty$. We say $\hat{\omega}(p)$ is obtained by an \emph{asymptotic $T$-method}.

We say that the asymptotic $T$-method is \emph{$\kappa$-catalytic} if in any inequality we have~$s \geq C n^\kappa$ for some constant~$C$.
\end{definition}

\begin{remark}
This class of methods works because each reduction $T^{\otimes k} \geq \left<n, n, m\right>^{\oplus s}$
gives an upper bound $\omega(q) \leq k \log_n r - \log_n s$ where $q = \log m \geq \log f(n) \to p$.
As the function $\omega(p)$ is continuous~\cite{DBLP:journals/tcs/LottiR83}, we get the required bound on $\omega(p)$ in the limit.
\end{remark}

\begin{remark}
The usual descriptions of the laser method applied to rectangular matrix multiplication
result in an asymptotic method because the construction involves an approximation of a certain probability distribution by a rational probability distribution.
As a result of this approximation, the matrix multiplication tensor constructed may have
format slightly smaller than $\left<n, n, n^p\right>$.
\end{remark}

\begin{theorem}
  \label{thm:approx}
 Any upper bound $\hat{\omega}(p)$ obtained by an
  asymptotic $T$-method satisfies %
  \begin{equation*}
    \hat{\omega}(p) \geq \frac{\log F(\left<2, 2, 2^p\right>) \log \asymprank(T)}{\log F(T)}
  \end{equation*}
  for every adequate $F$.
  
  For $\kappa$-catalytic methods
  \begin{equation*}
    \hat{\omega}(p) \geq \frac{\log F(\left<2, 2, 2^p\right>) \log \asymprank(T)}{\log F(T)} + \kappa \left(\frac{\log \asymprank(T)}{\log F(T)} - 1\right).
  \end{equation*}
\end{theorem}
\begin{proof}
  Suppose $T^{k} \geq \left<n, n, m\right>^{\oplus s}$. Then $\hat{\omega}_{k, s, n, m} = k \log_n \asymprank(T) - \log_n s$ is an upper bound on $\omega(p + o(1))$. Then, as in \autoref{thm:general}, we have
  \begin{equation*}
    \frac{\hat{\omega}_{k,s,n,m} + \log_n s}{\log_n F(\left<n, n, m\right>) + \log_n s} \geq \frac{\log \asymprank(T)}{\log F(T)}.
  \end{equation*}
  Because $F(T) \leq \asymprank(T)$, both fractions are greater than $1$ and for $0 \leq A \leq \log_n s$ it is true that
  \begin{equation*}
    \frac{\hat{\omega}_{k,s,n,m} + A}{\log_n F(\left<n, n, m\right>) + A} \geq \frac{\hat{\omega}(p)_{k,s,n,m} + \log_n s}{\log_n F(\left<n, n, m\right>) + \log_n s}.
  \end{equation*}
  As $n \to \infty$, we have $\log_n F(\left<n, n, m\right>) \geq \log F(\left<2, 2, 2^p\right>) + o(1)$ by \autoref{lem:adeq-cont}, and, if the method is $\kappa$-catalytic, then $\log_n s \geq \kappa + o(1)$.
  The upper bound $\hat{\omega}(p)$ given by the method is the limit $\lim \inf \hat{\omega}_{k,s,n,m}$.
  Taking $n \to \infty$, we get the required inequalities.
\end{proof}

\subsection{Barriers for mixed methods}\label{sec:mixed}

Coppersmith~\cite{DBLP:journals/jc/Coppersmith97} uses a combination of
Coppersmith--Winograd tensors of different format to get an upper bound on the rectangular
matrix multiplication exponent.
More specifically, he considers a sequence of tensors
$\CW_{7}^{\otimes 9n} \otimes \CW_{6}^{\otimes 8 \lfloor 0.6425 n \rfloor}$.
Our analysis applies to tensor sequences of this kind
because their asymptotic behaviour is similar to sequences of the form $T^{\otimes n}$ in the sense of the following two lemmas.

\begin{lemma}
  Let $S_1,S_2$ be tensors. Given functions $f_1, f_2 \colon \mathbb{N} \to \mathbb{N}$
  such that $f_i(n) = a_i n + o(n)$ for some positive real numbers $a_1, a_2$, define
  a sequence of tensors $T_n = S_1^{\otimes f_1(n)} \otimes S_2^{\otimes f_2(n)}$.
  Then for every adequate $F$ the sequence $\sqrt[n]{F(T_n)}$ is bounded from above.
\end{lemma}
\begin{proof}
  We have
  \begin{equation*}
    \sqrt[n]{F(T_n)} = \sqrt[n]{F(S_1^{\otimes f_1(n)} \otimes S_2^{\otimes f_2(n)})}
    \leq F(S_1)^{\frac{f_1(n)}{n}} F(S_2)^{\frac{f_2(n)}{n}}.
  \end{equation*}
  The right-hand side converges to $F(S_1)^{a_2} F(S_2)^{a_2}$ as $n\to\infty$ and, therefore, is bounded.
\end{proof}

\begin{lemma}
  Let $S_1,S_2$ be tensors. Given functions $f_1, f_2 \colon \mathbb{N} \to \mathbb{N}$
  such that $f_i(n) = a_i n + o(n)$ for some positive real numbers $a_1, a_2$, define
  a sequence of tensors $T_n = S_1^{\otimes f_1(n)} \otimes S_2^{\otimes f_2(n)}$.
  Then the sequence $\sqrt[n]{\asymprank(T_n)}$ converges.
\end{lemma}
\begin{proof}
  For this, we need Strassen's spectral characterization of the asymptotic rank~\cite{Strassen:88-asymptotic}.
  Strassen defines the asymptotic spectrum of tensors $X$
  as the set of all $\leq$-monotone, $\otimes$-multiplicative, $\oplus$-additive maps $\xi$ from tensors to positive reals such that $\xi(u \otimes v \otimes w) = 1$.
  Then~$X$ can be made into a compact (and Hausdorff) topological space such that the evaluation map $\xi \mapsto \xi(T)$ is continuous for all $T$, and
  \[
    \asymprank(T) = \max_{\xi \in X} \xi(T).
  \]
  For $\xi \in X$ we have
  \begin{equation*}
    \sqrt[n]{\xi(T_n)} = \sqrt[n]{\xi(S_1^{\otimes f_1(n)} \otimes S_2^{\otimes f_2(n)})}
    = \xi(S_1)^{\frac{f_1(n)}{n}} \xi(S_2)^{\frac{f_2(n)}{n}} \to \xi(S_1)^{a_1} \xi(S_2)^{a_2},
  \end{equation*}
  as $n\to \infty$.
  Because of compactness of $X$ this convergence is uniform in $\xi$.
  Therefore,
  \begin{equation*}
    \sqrt[n]{\asymprank(T_n)} = \sqrt[n]{\max_{\xi\in X} \xi(T_n)} \to \max_{\xi\in X} \xi(S_1)^{a_1} \xi(S_2)^{a_2},
  \end{equation*}
  as $n\to\infty$.
\end{proof}

\begin{definition}
  We call a sequence of tensors $\{T_n\}$ \emph{almost exponential} if the sequence
  $\sqrt[n]{\asymprank(T_n)}$ converges and $\sqrt[n]{F(T_n)}$ is bounded for each adequate $F$.
  We write $\asymprank(\{T_n\}) \coloneqq \lim_{n\to\infty} \sqrt[n]{\asymprank(T_n)}$
  and $F(\{T_n\}) \coloneqq \limsup_{n\to\infty} \sqrt[n]{F(T_n)}$.
\end{definition}

\begin{definition}[Asymptotic mixed method]
  Let $\{T_n\}$ be an almost exponential sequence of tensors with $\asymprank(\{T_n\}) \leq r$.
  Suppose we are given a collection of inequalities $\left<n, n, m\right>^{\oplus s} \leq T_k$
  where the values of $n$ are unbounded and $m \geq f(n)$ for some $f(n) = n^{p + o(1)}$.
  Then $\omega(p)$ is at most $\hat{\omega}(p) = \lim \inf \{k \log_n r - \log_n s\}$ where the
  limit is taken over all $k, n, s$ appearing in the collection of inequalities as $n \to \infty$.
  We say that $\hat{\omega}(p)$ is obtained by an \emph{asymptotic mixed $\{T_n\}$-method}.
  
  We say that the asymptotic mixed $\{T_n\}$-method is \emph{$\kappa$-catalytic} if in
  each inequality we have $s \geq C n^{\kappa}$ for some constant $C$.
\end{definition}

\begin{lemma}
  Asymptotic mixed methods give upper bounds on $\omega(p)$.
\end{lemma}
\begin{proof}
  Note that for a fixed tensor $T_k$ there are only a finite number of restrictions
  $\left<n,n,m\right>^{\oplus s} \leq T_k$ possible as the left tensor is of format
  $sn^2 \times snm \times snm$, which should be no greater than the format of $T_k$.
  Thus, because in an asymptotic mixed method the set of values of $n$ is unbounded,
  so is the set of values of $k$.

  For one restriction $\left<n, n, m\right>^{\oplus s} \leq T_k$ we have the inequality
  $s n^{\omega(\log_n m)} \leq \asymprank(T_k)$, that is, $\omega(\log_n m) \leq \log_n \asymprank(T_k) - \log_n s$.
  Since $\log_n m = p + o(1)$ and $\omega$ is a continuous function
  and $\asymprank(T_k) = (\asymprank(\{T_k\}) + o(1))^k$,
  we get in the limit the required inequality.
\end{proof}

\begin{theorem}
  \label{thm:mixed}
  Any upper bound $\hat{\omega}(p)$ obtained by an
  asymptotic mixed $\{T_n\}$-method satisfies %
  \begin{equation*}
    \hat{\omega}(p) \geq \frac{\log F(\left<2, 2, 2^p\right>) \log \asymprank(\{T_n\})}{\log F(\{T_n\})}
  \end{equation*}
  and for $\kappa$-catalytic methods,
  \begin{equation*}
    \hat{\omega}(p) \geq \frac{\log F(\left<2, 2, 2^p\right>) \log \asymprank(\{T_n\})}{\log F(\{T_n\})} + \kappa \left(\frac{\log \asymprank(\{T_n\})}{\log F(\{T_n\})} - 1\right).
  \end{equation*}
\end{theorem}
\begin{proof}
Recall that for a fixed $T_k$ the number of possible restrictions
  $\left<n,n,m\right>^{\oplus s} \leq T_k$ is finite, as the left-hand side tensor
  has format $sn^2 \times snm \times snm$, which should be no greater than that of $T_k$.
  Therefore, as $n$ tends to infinity, so does $k$.

  Consider now one restriction $\left<n,n,m\right>^{\oplus s} \leq T_k$.
  It gives the upper bound $\hat{\omega}_{k, s, n, m} \coloneqq \log_n \asymprank(T_k) - \log_n s$
  on $\omega(p + o(1))$.
  As in previous theorems, we have
  \begin{equation*}
    \frac{\hat{\omega}_{k,s,n,m} + \log_n s}{\log_n F(\left<n, n, m\right>) + \log_n s} \geq \frac{\log \asymprank(T_k)}{\log F(T_k)}
  \end{equation*}
  and
  \begin{equation*}
    \frac{\hat{\omega}_{k,s,n,m} + A}{\log_n F(\left<n, n, m\right>) + A} \geq \frac{\hat{\omega}(p)_{k,s,n,m} + \log_n s}{\log_n F(\left<n, n, m\right>) + \log_n s}
  \end{equation*}
  for any $A$ such that $0 \leq A \leq \log_n s$.
 
  Consider the behaviour of the involved quantities as $n$ and $k$ tend to infinity.
  Since $m \geq n^{p + o(1)}$, $\log_n F(\left<n, n, m\right>) \geq \log F(\left<2,2,2^p\right>) + o(1)$.
  For a catalytic method, we can choose $A = \kappa + o(1)$ such that $\log_n s \geq A$, and in general, we set $A = 0$.
  Since $\sqrt[k]{\asymprank(T_k)} = \asymprank(\{T_k\}) + o(1)$ and $\sqrt[k]{F(T_k)} \leq F(\{T_k\})+o(1)$,
  we have
  \begin{equation*}
    \frac{\log \asymprank(T_k)}{\log F(T_k)} \geq \frac{\log \asymprank(\{T_k\})}{\log F(\{T_k\})} + o(1).
  \end{equation*}
  And finally, $\lim \inf \hat{\omega}_{k,s,n,m}$ is $\hat{\omega}(p)$.
  In the limit, we get the required inequalities.
\end{proof}

\subsection{Barriers for the dual exponent $\alpha$}

Recall that there is an element $0 < p < 1$ such that $\omega(p) = 2$. We denote by $\alpha$ the largest such~$p$. We call $\alpha$ the dual exponent of matrix multiplication.
From the barrier theorem \autoref{thm:approx} for upper bounds on  $\omega(p)$ we can prove a barrier for lower bounds on $\alpha$, as we will now explain.

\begin{theorem}\label{th:barrier-alpha}
    For any $0 < p < 1$, if an asymptotic $T$-method proves that $p$ is a lower bound on $\alpha$, then
    \[
	p \leq \frac{2 \log F(T)}{\log \asymprank(T) \log F(\left<1,1,2\right>)} - \frac{\log F(\left<2,2,1\right>)}{\log F(\left<1,1,2\right>)}
    \]
    for all adequate $F$ such that $\log F(\left<1,1,2\right>) \neq 0$.
\end{theorem}
\begin{proof}
From the definition of $\alpha$ we see that, for any $0 < p < 1$, if an asymptotic $T$-method can prove that $p$ is a lower bound on $\alpha$, then it can prove the upper bound $\omega(p) \leq 2$.
Applying the barrier theorem (\autoref{thm:approx}), this implies
\[
  \frac{\log F(\langle 2, 2, 2^{p}\rangle) \log \asymprank(T)}{\log F(T)} \leq 2
\]
for all adequate $F$.
Using \autoref{quasi}, we obtain the claim. %
\end{proof}

\begin{remark}\label{alphaomega}
We note in passing that the matrix multiplication exponent $\omega$ and the dual exponent $\alpha$ are related via the inequality $\omega + \tfrac{\omega}{2}\alpha \leq 3$. Namely, from $\langle \ceil{n^\alpha}, n, n\rangle \leq \langle n^{2 + o(1)}\rangle$, $\langle n,\ceil{n^\alpha}, n\rangle \leq \langle n^{2 + o(1)}\rangle$ and $\langle n,n,\ceil{n^\alpha}\rangle \leq \langle n^{2 + o(1)}\rangle$ it follows that $\langle n^{2+\alpha}, n^{2+\alpha}, n^{2+\alpha}\rangle \leq \langle n^{6 + o(1)}\rangle$. Therefore, $\omega \leq 6/(2+\alpha)$, and the claim follows.
\end{remark}

\section{Numerical computation of barriers}\label{sec:numerical}

We will in this section discuss how to numerically evaluate the barrier of \autoref{thm:approx} and \autoref{th:barrier-alpha}. For this we will use the upper support functionals as adequate maps. We will compute some explicit values of barriers. 

\subsection{Upper support functionals}
Our main tool is a family of maps called the upper support functionals~\cite{strassen1991degeneration}. To define them, we will use the following notation. For $n \in \NN$ let $[n] \coloneqq \{1, 2, \ldots, n\}$. For any finite set $A$ let $\prob(A)$ be the set of probability vectors on $A$. For finite sets~$A_1, A_2, A_3$ and $P \in \prob(A_1 \times A_2 \times A_3)$ let $P_i \in \prob(A_i)$ be the $i$th marginal of~$P$ for $i \in [3]$. Let~$H(P)$ denote the Shannon entropy of~$P$.
Let $\FF^{n \times n \times n}$ be the set of 3-tensors of dimension $n \times n \times n$, viewed as 3-dimensional arrays. For~$T \in \FF^{n \times n \times n}$ let $\supp(T) \subseteq [n]^3$ be the support of $T$. %

Let $T \in \FF^{n \times n \times n}$. Let $\theta = (\theta_1, \theta_2, \theta_3) \in \prob([3])$. The \emph{upper support functional} is defined as
\begin{equation}
\zeta^\theta(T) = \min_{S \cong T} \max_{P \in \prob(\supp(S))} 2^{\sum_{i\in [3]} \theta_i H(P_i)}
\end{equation}
where $S$ goes over all tensors that can be obtained from $T$ by a basis transformation, that is, $S = (A,B,C)\cdot T$ where $A,B,C$ are invertible linear maps.

\begin{lemma}\label{rectsupp} %
$\zeta^\theta(\langle a,b,c\rangle) = a^{\theta_1 + \theta_3} b^{\theta_1 + \theta_2} c^{\theta_2 + \theta_3}$.
\end{lemma}
\begin{proof}
This is computed in \cite[Section 6]{strassen1991degeneration}.
\end{proof}

\begin{lemma}\label{basicprops}
The upper support functionals $\zeta^\theta$ are adequate.
\end{lemma}
\begin{proof}
The conditions of \autoref{def:adequate} are found in the following places.
(i), (ii), (iv) are in \cite[Theorem 2.8]{strassen1991degeneration}. (iii) follows from \autoref{rectsupp}. (v)~follows from the fact that every upper support functional is at most the maximum of the flattening ranks of the tensor \cite[page 135]{strassen1991degeneration}, and the flattening ranks lower bound the asymptotic rank.
\end{proof}

\begin{remark}More is known about the support functionals than \autoref{rectsupp} and \autoref{basicprops}. For example, they are multiplicative not only on the matrix multiplication tensors, but also on a larger family of tensors called oblique tensors \cite{strassen1991degeneration}.
\end{remark}

We obtain from \autoref{thm:approx} and \autoref{quasi} that
any upper bound $\hat{\omega}(p)$ on~$\omega(p)$ obtained by asymptotic $T$-methods must satisfy
  \begin{equation*}
    \hat{\omega}(p) \geq \frac{\log \zeta^\theta(\left<2, 2, 2^p\right>)}{\log \zeta^\theta(T)} \log \asymprank(T),
  \end{equation*}
which, using \autoref{rectsupp}, gives
\begin{equation}\label{eq:gen-bar}
\hat{\omega}(p) \geq
\max_\theta \frac{2 \theta_1 + \theta_3 + \theta_2 + p(\theta_2 + \theta_3)}{\log \zeta^\theta(T)} \log \asymprank(T).
\end{equation}

\subsection{Barriers for the Coppersmith--Winograd tensors} %

We know that $\asymprank(\CW_q) = q+2$. 
From \eqref{eq:gen-bar},
the barrier we get for $\CW_q$ is
\begin{align*}
\hat{\omega}(p) &\geq \max_\theta \frac{2\theta_1 + (p+1)(\theta_2 + \theta_3)}{\log_2 \zeta^\theta(\CW_q)} \log_2 \asymprank(\CW_q)\\ 
&\geq \max_\theta \frac{2\theta_1 + (p+1)(\theta_2 + \theta_3)}{\max_{P} \sum_{i=1}^3 \theta_i H(P_i)} \log_2 (q+2),
\end{align*}
where $P \in \prob(\supp(\CW_q))$ goes over all probability distributions on the support of $\CW_q$, which we recall is given by 
\[
\supp(\CW_q) = \{(i,i,0), (i,0,i), (0,i,i) : i \in [q]\} \cup \{(0,0,q+1), (0,q+1,0), (q+1,0,0)\}.
\]
This is easy to evaluate numerically, and we give explicit values in \autoref{subsec:values}.

\begin{remark}\label{rem:small-cw}
 In \autoref{subsec:values} we will for comparison, besides numerical values for the above barrier for $\CW_q$ also provide such values for the ``little'' Coppersmith--Winograd tensor $\cw_q$, which is the zero-one tensor with support $\supp(\cw_q) = \{(i,i,0), (i,0,i), (0,i,i) : i \in [q]\}.$ Unlike for $\CW_q$, the asymptotic rank of $\cw_q$ is not known and is between $q+1$ and $q+2$, and it is well-known that if $\asymprank(\cw_2) = 3$, then $\omega = 2$.
\end{remark}

\begin{remark}
We briefly discuss the standard method for using symmetry to simplify the computation of the support support functional $\zeta^\theta(\CW_q)$, which works similarly for other tensors with symmetry.
We are interested in computing
$\max_{P} \sum_{i=1}^3 \theta_i H(P_i)$,
where $P$ goes over all probability distributions on the support of $\CW_q$.
The symmetric group $S_q$ acts naturally on the support of $\CW_q$ by permuting the label set $[q]$. Suppose $P$ is a feasible point for the maximization. Then $\pi \cdot P$ for any $\pi \in S_q$ is feasible as well and has the same value. Thus the symmetrized point
$\frac{1}{\abs[0]{S_q}} \sum_{\pi \in S_q} \pi \cdot P$
is feasible and has at least the same value or better, by concavity of the Shannon entropy~$H$. We may thus assume that $P$ is constant on the six orbits of $\supp(\CW_q)$ under the action of~$S_q$, which are the sets $\{(i,i,0) : i \in [q]\}$, $\{(i,0,i) : i \in [q]\}$, $\{(0,i,i):i\in[q]\}$, $\{(0,0,q+1)\}$, $\{(0,q+1,0)\}$, and $\{(q+1,0,0)\}$. %

To make this concrete, let $P$ be the probability distribution that gives  probability $p_1$ to $(0,i,i)$, probability~$p_2$ to~$(i,0,i)$, probability~$p_3$ to $(i,i,0)$ and probability $r_1$ to $(q+1,0,0)$, probability~$r_2$ to $(0,q+1,0)$ and probability $r_3$ to $(0,0,q+1)$
where $p_1,p_2,p_3,r_1,r_2,r_3 \geq 0$ and $q p_1  + q p_2 + q p_3 + r_1 + r_2 + r_3 = 1$. The marginal probability vectors are
\begin{gather*}
P_1 = (q p_1 + r_2 + r_3, p_2 + p_3, \ldots,  p_2 + p_3, r_1)\\
P_2 = (q p_2 + r_1 + r_3, p_1 + p_3, \ldots,  p_1 + p_3, r_2)\\
P_3 = (q p_3 + r_1 + r_2, p_1 + p_2, \ldots,  p_1 + p_2, r_3).
\end{gather*}
By the grouping property of Shannon entropy, we have
\begin{gather*}
H(P_1) = (1-q p_1 - r_2 - r_3) (\log_2(q) + h(r_1)) + h(q p_1 + r_2 + r_3)\\
H(P_2) = (1-q p_2 - r_1 - r_3) (\log_2(q) + h(r_2)) + h(q p_2 + r_1 + r_3)\\
H(P_3) = (1-q p_3 - r_1 - r_2) (\log_2(q) + h(r_3)) + h(q p_3 + r_1 + r_2)
\end{gather*}
and
$\log_2 \zeta^\theta(\CW_q) \leq \max_{p_j, r_j} \sum_{i=1}^3 \theta_i H(P_i)$,
where $p_1, p_2, p_3, r_1, r_2, r_3 \geq 0$ and $qp_1 + qp_2 + qp_3 + r_1 + r_2 + r_3 = 1$.
\end{remark}

\subsection{Barriers for the dual exponent via Coppersmith--Winograd tensors}

From \autoref{th:barrier-alpha}, using the support functionals and its properties (\autoref{rectsupp}, \autoref{basicprops}), we get the following barrier for any lower bound $p$ on the dual exponent $\alpha$ via the intermediate tensor $\CW_q$,
\begin{align*}
p &\leq \min_\theta\, \biggl( \frac{2 \log_2 \zeta^\theta(\CW_q)}{\log_2 \asymprank(\CW_q) \log_2 \zeta^\theta(\left<1,1,2\right>)} - \frac{\log_2 \zeta^\theta(\left<2,2,1\right>)}{\log_2 \zeta^\theta(\left<1,1,2\right>)} \biggr)\\
&= \min_\theta\, \biggl(\frac{2 \max_P \sum_{i=1}^3 \theta_i H(P_i)}{\log_2(q+2) \cdot (\theta_2 + \theta_3)} - \frac{1 + \theta_1}{\theta_2 + \theta_3}\biggr).
\end{align*}
We give numerical evaluations in \autoref{subsec:values}.

\subsection{Some explicit values}\label{subsec:values}
As an illustration, we give in \autoref{tab1} the barriers for upper bounds on $\omega(2)$ via asymptotic $\CW_q$-methods for small $q$ by numerical optimization. We provide code to perform this optimization in \autoref{code}. Optimal values were obtained for $\theta$ with~$\theta_2 = \theta_3$. In \autoref{tab2} we give similar barriers for $\cw_q$. In \autoref{tab3} we give barriers for the dual exponent.

\begin{table}[H]
\centering
\begin{tabular}{lll}
\toprule
$q$ & $\hat{\omega}(2) \geq$ & $\theta_1$\\ %
\midrule
2& 3.0626& 0.096\\
3& 3.0726& 0.106\\
4& 3.0831& 0.116\\
5& 3.0936& 0.126\\
6& 3.1039& 0.136\\
7& 3.1138& 0.144\\
8& 3.1232& 0.152\\
9& 3.1323& 0.159\\
10& 3.1409& 0.165\\
11& 3.1491& 0.171\\
12& 3.1569& 0.176\\
13& 3.1643& 0.181\\
14& 3.1714& 0.185\\
\bottomrule
\end{tabular}
\caption{Barriers for upper bounds on $\omega(2)$ via asymptotic $\CW_q$-methods for small $q$.}
\label{tab1}
\end{table}

\begin{table}[H]
\centering
\begin{minipage}{3cm}
\begin{tabular}{lll}
\toprule
$q$ & $\hat{\omega}(2) \geq$ & $\theta_1$\\ %
\midrule
2& 3.7855& 0.000\\
3& 3.4828& 0.000\\
4& 3.3398& 0.000\\
5& 3.2582& 0.007\\
6& 3.2141& 0.054\\
7& 3.1920& 0.085\\
8& 3.1812& 0.107\\
9& 3.1767& 0.125\\
10& 3.1758& 0.138\\
11& 3.1772& 0.149\\
12& 3.1799& 0.158\\
13& 3.1835& 0.166\\
14& 3.1876& 0.173\\
\bottomrule
\end{tabular}
\end{minipage}
\hspace{1cm}
\begin{minipage}{3cm}
\begin{tabular}{lll}
\toprule
$q$ & $\hat{\omega}(2) \geq$ & $\theta_1$\\ %
\midrule
2& 3.0& 0.000\\
3& 3.0000& 0.000\\
4& 3.0000& 0.000\\
5& 3.0001& 0.007\\
6& 3.0077& 0.054\\
7& 3.0209& 0.085\\
8& 3.0357& 0.107\\
9& 3.0504& 0.125\\
10& 3.0646& 0.138\\
11& 3.0781& 0.149\\
12& 3.0906& 0.158\\
13& 3.1024& 0.166\\
14& 3.1134& 0.173\\
\bottomrule
\end{tabular}
\end{minipage}
\caption{Barriers for upper bounds on $\omega(2)$ via asymptotic $\cw_q$-methods for small $q$. On the left we assume that $\asymprank(\cw_q) = q+2$. On the right we assume only $\asymprank(\cw_q) \geq q+1$. Note that in that case, for $q = 2$, there is no barrier. (This is not surprising, since proving $\asymprank(\cw_2) = 3$ would imply~$\omega=2$). For $q = 3$ and $q=4$ it can be seen with more precision that the barrier is $>3$.}
\label{tab2}
\end{table}

\begin{table}[H]
\centering
\begin{tabular}{ll}
\toprule
$q$ & barrier on $\alpha$ \\ %
\midrule
2& 0.6218\\
3& 0.5998\\
4& 0.5777\\
5& 0.5583\\
6& 0.5408\\
7& 0.5259\\
8& 0.5129\\
9& 0.5001\\
10& 0.4914\\
11& 0.4772\\
12& 0.4692\\
13& 0.4614\\
14& 0.4529\\
\bottomrule
\end{tabular}
\caption{Barriers for the dual exponent $\alpha$ via $\CW_q$ for small $q$, using $\theta_1 = 0.999999$ and $\theta_2 = \theta_3$.}
\label{tab3}
\end{table}

\section*{Acknowledgements}
We thank Harold Nieuwboer for the code in \autoref{code}.

\vspace{0.5em}
\noindent
MC and VL were supported by VILLUM FONDEN via the QMATH Centre of Excellence under Grant No.~10059 and the European Research Council (Grant agreement No. 818761).

\vspace{0.5em}
\noindent
FLG was supported by JSPS KAKENHI grants Nos.~JP15H01677, JP16H01705, JP16H05853, JP19H04066 and by the MEXT Quantum Leap Flagship Program (MEXT Q-LEAP) grant No.~JPMXS0118067394.

\vspace{0.5em}
\noindent
JZ was supported by %
National Science Foundation under Grant No.~DMS-1638352 and Dutch Research Council (NWO) Veni grant VI.Veni.212.284..

\bibliographystyle{alphaurl}
\bibliography{barriers}

\newcommand{\etalchar}[1]{$^{#1}$}
\begin{thebibliography}{BCC{\etalchar{+}}17b}

\bibitem[ADW{\etalchar{+}}24]{alman2024asymmetryyieldsfastermatrix}
Josh Alman, Ran Duan, Virginia~Vassilevska Williams, Yinzhan Xu, Zixuan Xu, and Renfei Zhou.
\newblock More asymmetry yields faster matrix multiplication, 2024.
\newblock \href {https://arxiv.org/abs/2404.16349} {\path{arXiv:2404.16349}}.

\bibitem[AFLG15]{MR3388238}
Andris Ambainis, Yuval Filmus, and Fran\c{c}ois Le~Gall.
\newblock Fast matrix multiplication: limitations of the {C}oppersmith-{W}inograd method (extended abstract).
\newblock In {\em {P}roceedings of the 47th Annual {ACM} {S}ymposium on {T}heory of {C}omputing ({STOC} 2015)}, pages 585--593, 2015.
\newblock \href {https://arxiv.org/abs/1411.5414} {\path{arXiv:1411.5414}}, \href {https://doi.org/10.1145/2746539.2746554} {\path{doi:10.1145/2746539.2746554}}.

\bibitem[Alm19]{DBLP:conf/coco/Alman19}
Josh Alman.
\newblock Limits on the universal method for matrix multiplication.
\newblock In {\em Proceedings of the 34th Computational Complexity Conference ({CCC} 2019)}, pages 12:1--12:24, 2019.
\newblock \href {https://arxiv.org/abs/1812.08731} {\path{arXiv:1812.08731}}, \href {https://doi.org/10.4230/LIPIcs.CCC.2019.12} {\path{doi:10.4230/LIPIcs.CCC.2019.12}}.

\bibitem[AW18a]{alman_et_al:LIPIcs:2018:8360}
Josh Alman and Virginia~Vassilevska Williams.
\newblock {Further Limitations of the Known Approaches for Matrix Multiplication}.
\newblock In {\em Proceedings of the 9th Innovations in Theoretical Computer Science Conference (ITCS 2018)}, pages 25:1--25:15, 2018.
\newblock \href {https://arxiv.org/abs/1712.07246} {\path{arXiv:1712.07246}}, \href {https://doi.org/10.4230/LIPIcs.ITCS.2018.25} {\path{doi:10.4230/LIPIcs.ITCS.2018.25}}.

\bibitem[AW18b]{8555139}
Josh Alman and Virginia~Vassilevska Williams.
\newblock Limits on all known (and some unknown) approaches to matrix multiplication.
\newblock In {\em Proceedings of the 59th Annual IEEE Symposium on Foundations of Computer Science ({FOCS} 2018)}, pages 580--591, 2018.
\newblock \href {https://arxiv.org/abs/1810.08671} {\path{arXiv:1810.08671}}, \href {https://doi.org/10.1109/FOCS.2018.00061} {\path{doi:10.1109/FOCS.2018.00061}}.

\bibitem[AW21]{DBLP:conf/soda/AlmanW21}
Josh Alman and Virginia~Vassilevska Williams.
\newblock A refined laser method and faster matrix multiplication.
\newblock In {\em Proceedings of the 2021 {ACM-SIAM} Symposium on Discrete Algorithms ({SODA} 2021)}, pages 522--539, 2021.
\newblock \href {https://doi.org/10.1137/1.9781611976465.32} {\path{doi:10.1137/1.9781611976465.32}}.

\bibitem[BCC{\etalchar{+}}17a]{MR3631613}
Jonah Blasiak, Thomas Church, Henry Cohn, Joshua~A. Grochow, Eric Naslund, William~F. Sawin, and Chris Umans.
\newblock On cap sets and the group-theoretic approach to matrix multiplication.
\newblock {\em Discrete Anal.}, 2017.
\newblock \href {https://arxiv.org/abs/1605.06702} {\path{arXiv:1605.06702}}, \href {https://doi.org/10.19086/da.1245} {\path{doi:10.19086/da.1245}}.

\bibitem[BCC{\etalchar{+}}17b]{blasiak2017groups}
Jonah Blasiak, Thomas Church, Henry Cohn, Joshua~A Grochow, and Chris Umans.
\newblock Which groups are amenable to proving exponent two for matrix multiplication?
\newblock {\em arXiv}, 2017.
\newblock \href {https://arxiv.org/abs/1712.02302} {\path{arXiv:1712.02302}}.

\bibitem[BCS97]{burgisser1997algebraic}
Peter B{\"u}rgisser, Michael Clausen, and M.~Amin Shokrollahi.
\newblock {\em Algebraic complexity theory}, volume 315 of {\em Grundlehren Math. Wiss.}
\newblock Springer-Verlag, Berlin, 1997.
\newblock \href {https://doi.org/10.1007/978-3-662-03338-8} {\path{doi:10.1007/978-3-662-03338-8}}.

\bibitem[Bl{\"a}13]{blaser2013fast}
Markus Bl{\"a}ser.
\newblock {\em Fast Matrix Multiplication}.
\newblock Number~5 in Graduate Surveys. Theory of Computing Library, 2013.
\newblock \href {https://doi.org/10.4086/toc.gs.2013.005} {\path{doi:10.4086/toc.gs.2013.005}}.

\bibitem[CLS19]{10.1145/3313276.3316303}
Michael~B. Cohen, Yin~Tat Lee, and Zhao Song.
\newblock Solving linear programs in the current matrix multiplication time.
\newblock In {\em Proceedings of the 51st Annual ACM Symposium on Theory of Computing (STOC 2019)}, page 938–942, 2019.
\newblock \href {https://arxiv.org/abs/1810.07896} {\path{arXiv:1810.07896}}, \href {https://doi.org/10.1145/3313276.3316303} {\path{doi:10.1145/3313276.3316303}}.

\bibitem[Cop82]{DBLP:journals/siamcomp/Coppersmith82}
Don Coppersmith.
\newblock Rapid multiplication of rectangular matrices.
\newblock {\em {SIAM} J. Comput.}, 11(3):467--471, 1982.
\newblock \href {https://doi.org/10.1137/0211037} {\path{doi:10.1137/0211037}}.

\bibitem[Cop97]{DBLP:journals/jc/Coppersmith97}
Don Coppersmith.
\newblock Rectangular matrix multiplication revisited.
\newblock {\em J. Complexity}, 13(1):42--49, 1997.
\newblock \href {https://doi.org/10.1006/jcom.1997.0438} {\path{doi:10.1006/jcom.1997.0438}}.

\bibitem[CVZ18]{DBLP:conf/stoc/ChristandlVZ18}
Matthias Christandl, P{\'{e}}ter Vrana, and Jeroen Zuiddam.
\newblock Universal points in the asymptotic spectrum of tensors.
\newblock In {\em Proceedings of the 50th Annual {ACM} Symposium on Theory of Computing ({STOC} 2018)}, pages 289--296, 2018.
\newblock \href {https://arxiv.org/abs/1709.07851} {\path{arXiv:1709.07851}}, \href {https://doi.org/10.1145/3188745.3188766} {\path{doi:10.1145/3188745.3188766}}.

\bibitem[CVZ19]{DBLP:conf/coco/ChristandlVZ19}
Matthias Christandl, P{\'{e}}ter Vrana, and Jeroen Zuiddam.
\newblock Barriers for fast matrix multiplication from irreversibility.
\newblock In {\em Proceedings of the 34th Computational Complexity Conference ({CCC} 2019)}, pages 26:1--26:17, 2019.
\newblock \href {https://arxiv.org/abs/1812.06952} {\path{arXiv:1812.06952}}, \href {https://doi.org/10.4230/LIPIcs.CCC.2019.26} {\path{doi:10.4230/LIPIcs.CCC.2019.26}}.

\bibitem[CW90]{DBLP:journals/jsc/CoppersmithW90}
Don Coppersmith and Shmuel Winograd.
\newblock Matrix multiplication via arithmetic progressions.
\newblock {\em J. Symb. Comput.}, 9(3):251--280, 1990.
\newblock \href {https://doi.org/10.1016/S0747-7171(08)80013-2} {\path{doi:10.1016/S0747-7171(08)80013-2}}.

\bibitem[DB16]{diamond2016cvxpy}
Steven Diamond and Stephen Boyd.
\newblock {CVXPY}: {A} {P}ython-embedded modeling language for convex optimization.
\newblock {\em Journal of Machine Learning Research}, 17(83):1--5, 2016.
\newblock URL: \url{https://www.cvxpy.org}.

\bibitem[DWZ23]{DWZ23}
Ran Duan, Hongxun Wu, and Renfei Zhou.
\newblock Faster matrix multiplication via asymmetric hashing.
\newblock In {\em Proceedings of the 64th {IEEE} Annual Symposium on Foundations of Computer Science (FOCS 2023)}, pages 2129--2138, 2023.
\newblock \href {https://doi.org/10.1109/FOCS57990.2023.00130} {\path{doi:10.1109/FOCS57990.2023.00130}}.

\bibitem[HP98]{Huang+98}
Xiaohan Huang and Victor~Y. Pan.
\newblock Fast rectangular matrix multiplication and applications.
\newblock {\em J. Complexity}, 14(2):257--299, 1998.

\bibitem[KZHP08]{Ke+08}
ShanXue Ke, BenSheng Zeng, WenBao Han, and Victor~Y. Pan.
\newblock Fast rectangular matrix multiplication and some applications.
\newblock {\em Science in China Series A: Mathematics}, 51(3):389--406, 2008.

\bibitem[LG12]{le2012faster}
Fran{\c{c}}ois Le~Gall.
\newblock Faster algorithms for rectangular matrix multiplication.
\newblock In {\em Proceedings of the 53rd Annual IEEE Symposium on Foundations of Computer Science ({FOCS} 2012)}, pages 514--523, 2012.
\newblock \href {https://arxiv.org/abs/1204.1111} {\path{arXiv:1204.1111}}, \href {https://doi.org/10.1109/FOCS.2012.80} {\path{doi:10.1109/FOCS.2012.80}}.

\bibitem[LG14]{le2014powers}
Fran{\c{c}}ois Le~Gall.
\newblock Powers of tensors and fast matrix multiplication.
\newblock In {\em {P}roceedings of the 39th {I}nternational {S}ymposium on {S}ymbolic and {A}lgebraic {C}omputation ({ISSAC} 2014)}, pages 296--303, 2014.
\newblock \href {https://arxiv.org/abs/1401.7714} {\path{arXiv:1401.7714}}, \href {https://doi.org/10.1145/2608628.2608664} {\path{doi:10.1145/2608628.2608664}}.

\bibitem[LG24]{LG24}
Fran{\c{c}}ois Le~Gall.
\newblock Faster rectangular matrix multiplication by combination loss analysis.
\newblock In {\em Proceedings of the 2024 {ACM-SIAM} Symposium on Discrete Algorithms ({SODA} 2024)}, pages 3765--3791, 2024.
\newblock \href {https://doi.org/10.1137/1.9781611977912.133} {\path{doi:10.1137/1.9781611977912.133}}.

\bibitem[LR83]{DBLP:journals/tcs/LottiR83}
Grazia Lotti and Francesco Romani.
\newblock On the asymptotic complexity of rectangular matrix multiplication.
\newblock {\em Theor. Comput. Sci.}, 23:171--185, 1983.
\newblock \href {https://doi.org/10.1016/0304-3975(83)90054-3} {\path{doi:10.1016/0304-3975(83)90054-3}}.

\bibitem[LSZ19]{lee2019solving}
Yin~Tat Lee, Zhao Song, and Qiuyi Zhang.
\newblock Solving empirical risk minimization in the current matrix multiplication time.
\newblock In {\em Conference on Learning Theory ({COLT} 2019)}, pages 2140--2157, 2019.
\newblock URL: \url{http://proceedings.mlr.press/v99/lee19a.html}, \href {https://arxiv.org/abs/1905.04447} {\path{arXiv:1905.04447}}.

\bibitem[LU18]{DBLP:conf/soda/GallU18}
Fran\c{c}ois {Le Gall} and Florent Urrutia.
\newblock Improved rectangular matrix multiplication using powers of the {Coppersmith-Winograd} tensor.
\newblock In {\em Proceedings of the 29th Annual {ACM-SIAM} Symposium on Discrete Algorithms ({SODA} 2018)}, pages 1029--1046, 2018.
\newblock \href {https://arxiv.org/abs/1708.05622} {\path{arXiv:1708.05622}}, \href {https://doi.org/10.1137/1.9781611975031.67} {\path{doi:10.1137/1.9781611975031.67}}.

\bibitem[Sto10]{stothers2010complexity}
Andrew~James Stothers.
\newblock {\em On the complexity of matrix multiplication}.
\newblock PhD thesis, University of Edinburgh, 2010.
\newblock \url{http://hdl.handle.net/1842/4734}.

\bibitem[Str69]{strassen-gaussian-1969}
Volker Strassen.
\newblock Gaussian elimination is not optimal.
\newblock {\em Numerische Mathematik}, 13(4):354--356, 1969.
\newblock \href {https://doi.org/10.1007/BF02165411} {\path{doi:10.1007/BF02165411}}.

\bibitem[Str87]{Strassen:87-relative}
Volker Strassen.
\newblock Relative bilinear complexity and matrix multiplication.
\newblock {\em J. reine angew. Math.}, 375/376:406--443, 1987.
\newblock \href {https://doi.org/10.1515/crll.1987.375-376.406} {\path{doi:10.1515/crll.1987.375-376.406}}.

\bibitem[Str88]{Strassen:88-asymptotic}
Volker Strassen.
\newblock The asymptotic spectrum of tensors.
\newblock {\em J. reine angew. Math.}, 384:102--152, 1988.
\newblock \href {https://doi.org/10.1515/crll.1988.384.102} {\path{doi:10.1515/crll.1988.384.102}}.

\bibitem[Str91]{strassen1991degeneration}
Volker Strassen.
\newblock Degeneration and complexity of bilinear maps: some asymptotic spectra.
\newblock {\em J. Reine Angew. Math.}, 413:127--180, 1991.
\newblock \href {https://doi.org/10.1515/crll.1991.413.127} {\path{doi:10.1515/crll.1991.413.127}}.

\bibitem[vdB20]{DBLP:conf/soda/Brand20}
Jan van~den Brand.
\newblock A deterministic linear program solver in current matrix multiplication time.
\newblock In {\em Proceedings of the 2020 {ACM-SIAM} Symposium on Discrete Algorithms ({SODA} 2020)}, pages 259--278, 2020.
\newblock \href {https://arxiv.org/abs/1910.11957} {\path{arXiv:1910.11957}}, \href {https://doi.org/10.1137/1.9781611975994.16} {\path{doi:10.1137/1.9781611975994.16}}.

\bibitem[Wil12]{MR2961552}
Virginia~Vassilevska Williams.
\newblock Multiplying matrices faster than {C}oppersmith-{W}inograd (extended abstract).
\newblock In {\em Proceedings of the 44th Annual {ACM} {S}ymposium on {T}heory of {C}omputing ({STOC} 2012)}, pages 887--898, 2012.
\newblock \href {https://doi.org/10.1145/2213977.2214056} {\path{doi:10.1145/2213977.2214056}}.

\bibitem[WXXZ24]{DBLP:conf/soda/WilliamsXXZ24}
Virginia~Vassilevska Williams, Yinzhan Xu, Zixuan Xu, and Renfei Zhou.
\newblock New bounds for matrix multiplication: from alpha to omega.
\newblock In {\em Proceedings of the 2024 {ACM-SIAM} Symposium on Discrete Algorithms ({SODA} 2024)}. {SIAM}, 2024.
\newblock \href {https://doi.org/10.1137/1.9781611977912.134} {\path{doi:10.1137/1.9781611977912.134}}.

\end{thebibliography}

\newpage
\appendix
\section{Python code}\label{code}
The Python program below computes the values given in \autoref{tab1}, \autoref{tab2} and \autoref{tab3}. This uses the convex optimization package \verb!cvxpy!~\cite{diamond2016cvxpy}.

\begin{verbatim}
import numpy as np
import cvxpy as cp #we used version 1.4.2
cpz = cp.Constant(0)

def log_support_functional(shape, support, theta):
    n = len(shape)
    var = {w: cp.Variable() for w in support}
    constraints = []
    for v in var.values():
        constraints.append(v >= 0)
    constraints.append(sum(v for v in var.values()) == 1)

    entropies = []
    pdict = dict()
    for i in range(n):
        ent_i = cpz
        for j in range(shape[i]):
            pij = sum([v for (w, v) in var.items() if w[i] == j], start=cpz)
            pdict[(i,j)] = pij
            ent_i += cp.entr(pij)
        entropies.append(ent_i)
    objective = 1/np.log(2) * sum(theta[i] * entropies[i] for i in range(n))
    problem = cp.Problem(cp.Maximize(objective), constraints)
    problem.solve() 
    
    return problem.value


# Table 1: CW

p = 2
def f(q, t1):
    t2 = (1.0 - t1)/2
    t3 = 1.0 - t1 - t2
    s = [(0,i,i) for i in range(1, q+1)] \
      + [(i,0,i) for i in range(1, q+1)] \
      + [(i,i,0) for i in range(1, q+1)] \
      + [(0,0,q+1), (0,q+1,0),(q+1,0,0)]
    v = (2 * t1 + (p+1)*(t2 + t3)) \
      / log_support_functional((q+2,q+2,q+2), s, (t1, t2, t3)) * np.log2(q+2)
    return v

for q in range(2, 15):
    m = max([(f(q, 0.001 * t1), 0.001 * t1) for t1 in range(1,201)])
    print(q, str(m[0])[:6], str(m[1])[:5])
\end{verbatim}

\newpage
\begin{verbatim}
# Table 2 right: cw

p = 2
def f(q, t1):
    t2 = (1.0 - t1)/2
    t3 = 1.0 - t1 - t2
    s = [(0,i,i) for i in range(1, q+1)] \
      + [(i,0,i) for i in range(1, q+1)] \
      + [(i,i,0) for i in range(1, q+1)]
    v = (2 * t1 + (p+1)*(t2 + t3)) \
      / log_support_functional((q+1,q+1,q+1), s, (t1, t2, t3)) * np.log2(q+1)
    return v

for q in range(2, 15):
    m = max([(f(q, 0.001 * t1), 0.001 * t1) for t1 in range(0,201)])
    print(q, str(m[0])[:6], str(m[1])[:5])
\end{verbatim}

\begin{verbatim}

# Table 2 left: cw using best-known upper bound on asymptotic rank

p = 2
def f(q, t1):
    t2 = (1.0 - t1)/2
    t3 = 1.0 - t1 - t2
    s = [(0,i,i) for i in range(1, q+1)] \
      + [(i,0,i) for i in range(1, q+1)] \
      + [(i,i,0) for i in range(1, q+1)]
    v = (2 * t1 + (p+1)*(t2 + t3)) \
      / log_support_functional((q+1,q+1,q+1), s, (t1, t2, t3)) * np.log2(q+2)
    return v

for q in range(2, 15):
    m = max([(f(q, 0.001 * t1), 0.001 * t1) for t1 in range(0,201)])
    print(q, str(m[0])[:6], str(m[1])[:5])
\end{verbatim}

\begin{verbatim}

# Table 3: dual exponent

for q in range(2,15):
    s = [(0,i,i) for i in range(1, q+1)] \
      + [(i,0,i) for i in range(1, q+1)] \
      + [(i,i,0) for i in range(1, q+1)] \
      + [(0,0,q+1), (0,q+1,0),(q+1,0,0)]
    def g(t1):
        t2 = (1 - t1)/2
        t3 = 1 - t1 - t2
        v = 2*log_support_functional((q+2,q+2,q+2), s, (t1, t2, t3)) \
          / (np.log2(q+2)*(t2+t3))  - (2 * t1 + t2 + t3) / (t2 + t3)
        return v

    print(q, str(N(ceil(g(0.999999)*10000)/10000))[:6])
\end{verbatim}

\newpage

\noindent
\textbf{Matthias Christandl}\\
University of Copenhagen\\
Universitetsparken 5, 2100 Copenhagen Ø, Denmark\\
Email: \href{mailto:christandl@math.ku.dk}{christandl@math.ku.dk}\\[1em]
\textbf{François Le Gall}\\
Nagoya University\\
Furocho, Chikusaku, Nagoya Aichi 464-8602, Japan\\
Email: \href{mailto:legall@math.nagoya-u.ac.jp}{legall@math.nagoya-u.ac.jp}\\[1em]
\textbf{Vladimir Lysikov}\\
Ruhr University Bochum\\
Universitätsstraße 150, 44801 Bochum, Germany\\
Email: \href{mailto:vladimir.lysikov@rub.de}{vladimir.lysikov@rub.de}\\[1em]
\textbf{Jeroen Zuiddam}\\
University of Amsterdam\\ Science Park 107, Amsterdam, Netherlands\\
Email: \href{mailto:j.zuiddam@uva.nl}{j.zuiddam@uva.nl}

\end{document}